\documentclass[pdflatex,sn-mathphys]{sn-jnl}

\usepackage[utf8]{inputenc}
\usepackage[T1]{fontenc}
\usepackage{amsmath,amsfonts,amsthm,amssymb}

\usepackage{xfrac}

\usepackage{xspace}
\usepackage{enumerate}
\usepackage{multicol}
\usepackage{lmodern}

\usepackage[algoruled, vlined,algo2e,linesnumbered]{algorithm2e}

\usepackage{setspace}
\usepackage{thmtools} 					

\usepackage{footnote}
\makesavenoteenv{tabular}
\makesavenoteenv{table}

\usepackage{adjustbox}
\newcommand{\zell}[2][c]{ \begin{tabular}[#1]{@{}c@{}}#2\end{tabular}}

\usepackage{hyperref}
\usepackage
[
    noabbrev,   
    nameinlink, 
]
{cleveref} 

\crefname{case}{case}{cases}
\creflabelformat{stat}{#2{\upshape#1}#3}

\crefname{ineq}{inequality}{inequalities}
\creflabelformat{ineq}{#2{(\upshape#1)}#3} 

\crefname{alg}{algorithm}{algorithms}
\creflabelformat{alg}{#2{#1}#3}

\normalbaroutside

\newcommand*{\Otilde}{\widetilde{O}}

\newcommand*{\D}{\mathcal{D}}

\renewcommand*{\L}{\mathcal{L}}
\renewcommand*{\P}{\mathcal{P}}
\renewcommand*{\T}{\mathcal{T}}

\newcommand*{\poly}{\textsf{poly}}
\newcommand*{\polylog}{\textsf{polylog}}

\newcommand*{\nwspace}{\hspace*{.1em}} 

\DeclareMathOperator{\diam}{diam}
\DeclareMathOperator*{\argmax}{arg\,max}

\makeatletter
\DeclareRobustCommand{\cev}[1]{%
  {\mathpalette\do@cev{#1}}%
}
\newcommand{\do@cev}[2]{%
  \vbox{\offinterlineskip
    \sbox\z@{$\m@th#1 x$}%
    \ialign{##\cr
      \hidewidth\reflectbox{$\m@th#1\vec{}\mkern4mu$}\hidewidth\cr
      \noalign{\kern-\ht\z@}
      $\m@th#1#2$\cr
    }%
  }%
}
\makeatother

\let\oldsqrt\sqrt
\def\hksqrt{\mathpalette\DHLhksqrt}
\def\DHLhksqrt#1#2{\setbox0=\hbox{$#1\oldsqrt{#2\,}$}\dimen0=\ht0
   \advance\dimen0-0.2\ht0
   \setbox2=\hbox{\vrule height\ht0 depth -\dimen0}%
   {\box0\lower0.4pt\box2}}
\renewcommand\sqrt\hksqrt

\renewcommand{\leq}{\leqslant}
\renewcommand{\geq}{\geqslant}
\renewcommand{\le}{\leqslant}
\renewcommand{\ge}{\geqslant}
\newcommand{\eps}{\varepsilon}

\usepackage{tikz}

\newcommand*{\trlength}{1.2}
\newcommand*{\noderadius}{6.5pt}
\usetikzlibrary{arrows.meta, quotes, calc, fit, decorations.pathreplacing}
\tikzset{
node distance={\trlength cm}, vert/.style = {draw, circle, inner sep = 0 pt, minimum size = 2 * \noderadius}, every path/.style = {}, every label/.append style={rectangle, font = {\normalsize}}, fac/.style = {circle,fill,inner sep=1.5pt}, every node/.style = {font = {\normalsize}}, every edge quotes/.style = {auto, font = {\normalsize}, inner sep = 2pt}
}

\definecolor{cred}{HTML}{D81B60}
\definecolor{cblue}{HTML}{1E88E5}
\definecolor{cyellow}{HTML}{D09C00}
\definecolor{cgreen}{HTML}{5B8600}

\jyear{2026}%

\theoremstyle{thmstyleone}%
\newtheorem{theorem}{Theorem}
\newtheorem{lemma}[theorem]{Lemma}

\newtheorem*{question*}{Question}

\providecommand{\ignore}[1]{}

\begin{document}

\title[Fault-Tolerant ST-Diameter Oracles]{Fault-Tolerant ST-Diameter Oracles\footnote[2]{%
	An extended abstract of this work appeared at ICALP 2023~\cite{Bilo23STdiameter}.}}

\author[1]{\fnm{Davide} \sur{Bilò}}\email{davide.bilo@univaq.it}
\author[2]{\fnm{Keerti} \sur{Choudhary}}\email{keerti@iitd.ac.in}
\author[3]{\fnm{Sarel} \sur{Cohen}}\email{sarel.cohen@runi.ac.il}
\author[4]{\fnm{Tobias} \sur{Friedrich}}\email{tobias.friedrich@hpi.de}
\author[4]{\fnm{Simon} \sur{Krogmann}}\email{simon.krogmann@hpi.de}
\author*[5]{\fnm{Martin} \sur{Schirneck}}\email{martin.schirneck@kit.edu}

\affil[1]{\orgdiv{Department of Information Engineering, Computer Science and Mathematics},
\orgname{University of L'Aquila}, \country{Italy}}

\affil[2]{\orgdiv{Department of Computer Science and Engineering},
\orgname{Indian Institute of Technology Delhi}, \country{India}}

\affil[3]{\orgdiv{Efi Arazi School of Computer Science},
\orgname{Reichman University}, \country{Israel}}

\affil[4]{\orgdiv{Hasso Plattner Institute},
\orgname{University of Potsdam}, \country{Germany}}

\affil[5]{\orgdiv{Derpartment of Informatics},
\orgname{Karlsruhe Institute of Technology}, \country{Germany}}

\abstract{%
	Given two vertex sets $S$ and $T$ in a graph, the $ST$-diameter is
	the maximum $s$-$t$-distance between vertices $s \in S$ and $t \in T$.
	We study the problem of estimating the $ST$-diameter of graphs 
	that are subject to a small number of transient edge failures.
	An \emph{$f$-edge fault-tolerant $ST$-diameter oracle} ($f$-FDO-$ST$) is a data structure that preprocesses a graph $G$, sets $S$, $T$, and a positive integer $f$. 
	When queried with a set $F$ of at most $f$ failing edges, the oracle returns an estimate $\widehat{D}$ of the $ST$-diameter in $G\,{-}\,F$. 
	The oracle is said to have stretch $\sigma \geq 1$ if $\diam(G{-}F,S,T) \leq \widehat{D} \leq \sigma \cdot \diam(G{-}F,S,T)$.
	\vspace*{.5em}

	We design new $f$-FDO-$ST$s by reducing their construction
	to that of all-pairs and single-source \emph{distance sensitivity oracles} ($f$-DSOs).
	These are data structures that estimate the pairwise graph distances,
	or respectively the distances from a distinguished source, under up to $f$ failures.
	We obtain several new trade-offs between the size of the $ST$-diameter oracles,
	their stretch guarantees, query and preprocessing times 
	by combining our black-box reductions with $f$-DSO results from the literature.
	\vspace*{.5em}

	We fruther provide a lower bound on the space requirement of approximate 
	$ST$-diameter oracles. 
	We prove that there exists a family of graphs for which any $f$-FDO-$ST$ 
	with sensitivity $f \ge 2$ and stretch better than $5/3$ 
	requires $\Omega(n^{3/2})$ bits of space, regardless of the query time.
}

\keywords{diameter oracles, distance sensitivity oracles, space lower bounds, fault-tolerant data structures}

\maketitle

\section{Introduction}
\label{sec:intro}

The diameter of a graph is the largest distance between any two of its vertices.
It is one of the most fundamental graph parameters
as it measures how fast information can spread in a network
or how quickly any other node can be visited.
The problem of computing, or at least approximating, the diameter in a time-efficient manner has been extensively studied, see e.g.~\cite{Aingworth,AnconaHRWW19,BackursRSWW21,Chechik:2014,ChoudharyG20,CrescenziGLM12,Roditty16ApproxDiamEncycAlg,RodittyW:2013,TakesK11}.
We continue this investigation through the lens of fault tolerance,
assuming that the graph undergoes a small number of transient edge failures.
The interest in this setting stems from the fact that most real-world networks are prone to errors. 
These failures, though unpredictable, are temporary due to some simultaneous repair process that is 
undertaken in these applications.
This has motivated research on the design of 
$f$-\emph{edge fault-tolerant oracles}, 
which are a compact data structures that can quickly report the desired solution or graph property after up to $f$ links failed in the network.
The \emph{sensitivity} parameter $f$ describes the oracle's robustness against failures.
In the last two decades, a lot of work went into constructing fault-tolerant
oracles for graph problems
with a special focus on connectivity~\cite{ChakrabortyC20,DuanP09a,DuanP10}
and finding shortest paths~\cite{AlonChechikCohen19CombinatorialRP,BCGLPP18,%
GrandoniVWilliamsFasterRPandDSO_journal,BaswanaCHR20,BaKa_FOCS06,%
BiloG0P22Algorithmica,CLPR10,DeThChRa08,DuanGR21,DuRe22,WY13}. 

The landscape for extremal distances like the diameter is far less explored. 
We are only aware of three articles in that direction.
The problem of designing fault-tolerant diameter oracles was originally raised by 
Henzinger, Lincoln, Neumann, and Vassilevska Williams~\cite{HenzingerL0W17}.
It has recently been studied by Bilò, Cohen, Friedrich, and Schirneck~\cite{BCFS21DiameterOracle_MFCS} and the same authors together with Choudhary~\cite{Bilo22Extremal}.
For a more detailed discussion, see \Cref{subsec:relatedWork-fdos}. 

We broaden the scope to the notion of $ST$-diameter 
that was introduced and studied in recent years.
For a given graph $G = (V,E)$ and two sets $S,T \subseteq V$ of vertices,
the $ST$-\emph{diameter} $\diam(G,S,T) = \max_{s \in S, t \in T} d_{G}(s,t)$
is the maximum distance between vertices in $S$ and $T$.
Clearly, when choosing $S = T = V$, we recover the (ordinary) graph diameter.
All previous works on the $ST$-diameter focus on static graphs.
For example, Backurs, Roditty, Segal, Vassilevska Williams, and Wein~\cite{BackursRSWW21} proved that for any undirected graph one can deterministically compute a $3$-approximation of the $ST$-diameter
in time $O(m)$.
They also provided a randomized $2$-approximate algorithm in time $\Otilde(m\sqrt{n})$ time.\footnote{%
	For a non-negative function $g(n,m,f)$, we write $\Otilde(g)$ for $O(g \cdot \polylog(n))$.
} 
Dalirrooyfard, Vassilevska Williams, Vyas, and Wein~\cite{DalirrooyfardW019a} investigated the problem of computing the bi-chromatic $ST$-diameter, 
the special case where the sets $S$ and~$T$ partition $V$.

When applying the fault-tolerant model to the $ST$-diameter, we get the following notions. 
An \emph{$f$-edge fault-tolerant $ST$-diameter oracle} ($f$-FDO-$ST$) is a data structure that
stores information about the input graph $G$ in a preprocessing step. 
When queried with a set $F \subseteq E$ of at most $f$ edges, the oracle returns 
an estimate $\widehat{D}$ of the $ST$-diameter of the graph $G\,{-}\,F$.
This is the maximum $s$-$t$-distances, for $s \in S$ and $t \in T$,
under the condition that a shortest path realizing this distance
cannot use a failing edge from $F$.
For the (ordinary) diameter, that is, when $S=T=V$, the data structure is called an \emph{$f$-edge fault-tolerant diameter oracle} ($f$-FDO).
\Cref{fig:problem} illustrates this problem,
we use this graph as a running example.
We say an oracle has \emph{stretch} $\sigma \geq 1$ if the value $\widehat{D}$ returned upon query $F$ satisfies $\diam(G{-}F,S,T) \leq \widehat{D} \leq \sigma \cdot \diam(G{-}F,S,T)$.
Other important performance parameters of such a data structure are its \emph{query time}
as well as its \emph{space} requirement.\footnote{
	Unless otherwise stated, the space of a data structure is given in the $O(\log n)$-bit word RAM model
	by the number of stored machine words.
}
We also consider the time needed to \emph{preprocess} the oracle as a secondary measure.

\begin{figure}[t]
	\centering
	\includegraphics[scale=1,page=2]{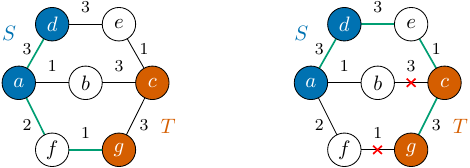}
	\caption{The edge-weighted graph we use as a running example. 
	It has unique shortest paths and distinguished vertex sets $S = \{a,d\}$ and $T = \{c,g\}$.
	The \mbox{$ST$-diameter} is $6$, realized by the shortest path from vertex $d$ to $g$ (in green).
	If the edges $\{b,c\}$ and $\{f,g\}$ fail (red crosses), the \mbox{$ST$-diameter} lengthens to $10$.
	When queried with the two edges,
	an $f$-edge fault-tolerant $ST$-diameter oracle with stretch $\sigma = 2$
	returns a value between $10$ and $20$.}
\label{fig:problem}
\end{figure}

Since there is a plethora of work on fault-tolerant oracles 
for shortest paths and the $ST$-diameter is somewhat  related to that,
it is a natural question whether we can convert the results on distance computation under failures into 
oracles for the $ST$-diameter without sacrificing too much performance.

\vspace*{-1em}
\begin{question*}
	Are there black-box reductions from fault-tolerant $ST$-diameter oracles
	to distance oracles without considerable overhead in stretch, query time, and space?
\end{question*}
\vspace*{-1em}

The epitome of fault-tolerant data structures for shortest-path distances are 
\emph{$f$-edge fault-tolerant distance sensitivity oracles} ($f$-DSOs).
The more common \emph{all-pairs} variety can be queried with a triplet $(s,t,F)$ where $s$, $t$ are any two vertices
and $F$ is a set of at most $f$ edges.
The oracle then returns an estimate of the \emph{replacement distance} $d_{G-F}(s,t)$
between $s$ and $t$.
A \emph{single-source} $f$-DSO has a designated source vertex $s$
and the queries only specify the target $t$ and failures $F$.
If the variant is not explicitly stated, we mean all pairs.

When using distance sensitivity oracles to build data structures for the $ST$-diameter
it is important to insist on efficiency.
Consider access to some $f$-DSO $\D$.
Naively, whenever the $ST$-diameter oracle receives a query set $F$, it could just invoke $\D$ for all
pairs of vertices from $S \nwspace{\times}\nwspace T$ and report their maximum replacement distance.
The resulting $|S||T| \nwspace \mathtt{Q}$ query time, where $\mathtt{Q}$ denotes the query time of $\D$,
is prohibitive if $S$ and $T$ are large.
Instead, all our query times are bounded by a function of the sensitivity $f$ 
and query time $\mathtt{Q}$ only.

We show how to employ both kinds of DSOs 
in the construction of fault-tolerant $ST$-diameter oracles.
All our reductions are oblivious to (positive) edge weights, 
but most results assume the input graph $G$ to be undirected.

\begin{table}[t]
\setstretch{1.2}
\centering
\caption{Properties of the fault-tolerant $ST$-diameter oracles for undirected graphs 
obtained via \Cref{thm:st-diameter-where-st-are-sets} using all-pairs distance sensitivity oracles
from the literature.
The maximum edge weight for graphs with arbitrary positive weights is denoted $W$,
for integer-weighted graphs it is $M$.
The parameter $k \ge 1$ is a positive integer, $\varepsilon > 0$ a positive real, 
and $\alpha \in [0,1]$ is a real number in the unit interval.
We use $c > 1$ for a constant, and $\omega < 2.371339$ for the matrix multiplication exponent.}
\label{table:FDO-ST_from_all-pairs_DSO}
\begin{adjustbox}{max width=\textwidth}
\begin{tabular}{ccccc}
\noalign{\hrule height 1pt}

\bf Sensitivity & \bf Stretch & \bf Space & \bf Query time & 
\bf Ref. \\
\noalign{\hrule height 1pt}\\[-10pt]

1 & 4& $\Otilde(n^2)$ & $O(1)$ & \cite{BeKa08,BeKa09}\\[.5em]

1 & 4& $\Otilde(n^2)$ & $O(1)$ & \cite{KarczmarzSankowski23SensitivityandDynamicDOs}\\[.5em]

1 & $(6k{-}3)(1{+}\eps)+1$ & 
$\Otilde(n^{3/2}+k^5 n^{1+1/k}/\eps^{4})$
& $O(1)$ & \cite{BaswanaK13}\vspace*{.5em}\\
\noalign{\hrule height 0.25pt}\\[-10pt]

2 & 4& $\Otilde(n^2)$ & 
$\Otilde(1)$ & \cite{DuanP09a}\vspace*{.25em}\\
\noalign{\hrule height 0.25pt}\\[-10pt]

$f=O(1)$ & 
$10+\varepsilon$ & 
$\Otilde(n^{2-\frac{\alpha}{2(f+1)}} (\log(n)/\varepsilon)^{O(f)})$ 
& $O(n^{\frac{\alpha}{2}}/\varepsilon^2)$ 
& \cite{BCCCFKS24TheoretiCS}\vspace*{.5em}\\
\noalign{\hrule height 0.25pt}\\[-10pt]



$f = o(\frac{\log n}{\log \log n})$ & 
4 & 
$\Otilde(n^{3-\alpha})$ & 
$\Otilde(f^2n^{2-(1-\alpha)/f})$
& \cite{WY13}\\[.5em]

$f = o(\frac{\log n}{\log \log n})$ & 
 $4+\varepsilon$ & 
$O(f n^{2+o(1)} \log(W)/\varepsilon^{f})$ & 
$\Otilde(f^7 \log \log W)$
& \cite{ChCoFiKa17}\vspace*{.5em}\\
\noalign{\hrule height 0.25pt}\\[-10pt]

$f \geq 1$ & 
4 & 
$O(fn^4)$ & 
$f^{O(f)}$ & \cite{DuRe22}\\[.5em] 

$f \ge 1$ &  
4 & 
$O(f^4n^2\log(Mn))$ &  
$O((c f \log(Mn))^{O(f^2)})$ &  
\cite{DeyGupta24}\\[.5em]

$f \geq 1$ & 
4 & 
$O(Mn^{2+\alpha})$ & 
$\Otilde(f^4 Mn^{2-\alpha} + f^{\omega+2} Mn)$
& \cite{BrSa19}\\[.5em]

$f \geq 1$ &  $(24k{-}6)(f{+}1)+1$ &
$O(n^{3/2+o(1)}+ fkn^{1+1/k}\log{(Wn)})$
&
$\Otilde( 2^f \nwspace f^2)$
& \cite{CLPR10}\\[.5em]

$f \geq 1$ &  $(24k{-}6)(f{+}1)+1$ &
$O(n^{2}+ fkn^{1+1/k}\log{(Wn)})$
&
$\Otilde(f^3)$
& \cite{CLPR10}
\end{tabular}
\end{adjustbox}
\end{table}

\begin{restatable}{theorem}{stdiameterwherestaresets}
\label{thm:st-diameter-where-st-are-sets}
	Let $G=(V,E)$ be an undirected positively edge-weighted graph with $n$ vertices and $m$ edges.
	Let $S,T \subseteq V$ be two non-empty sets. 
	Assume access to an \mbox{$f$-DSO} for $G$ with stretch $\sigma$, 
	space $\mathtt{S}$, query time $\mathtt{Q}$, and preprocessing time $\mathtt{P}$.
	\begin{enumerate}[(i)]
		\item There exists an $f$-FDO-$ST$ for $G$ with stretch $1 + 3\sigma$,
			space $\mathtt{S} + O(n^2)$, query time $O(f^2(\log(f)+\mathtt{Q}))$,
			and preprocessing time \mbox{$\mathtt{P}+ \Otilde(mn + n|S||T|)$}.
		\item If $f \le \log_2(n)$, the space can be improved to
			$\mathtt{S} + O(2^{f/2}\nwspace n^{3/2} \sqrt{\log n})$
			at the\\[.25em] expense of an $O(f^2 \nwspace (2^f +\mathtt{Q}))$ query time.
			The stretch and preprocessing time remain the same.
	\end{enumerate}
	If additionally $M\,{=}\, \max(|S|,|T|) \,{\le}\, 2^{f-1}$ holds, the space of the second variant further decreases to
	$\mathtt{S} + O(M^{1/2} \nwspace n^{3/2} \sqrt{\log n})$,    
    with a query time of $O(f^2 \nwspace (M +\mathtt{Q}))$.
\end{restatable}

The result offers two constructions
that mainly differ in the space overhead (over the size $\mathtt{S}$ of the underlying $f$-DSO) and the query time.
Some remark on the preprocessing time may be in order.
The reduction itself takes time $\mathtt{P}+ O(mn + n^2 \log n + n|S||T|)$ to compute,
but for technical reasons requires that the shortest paths in $G$ are unique.\footnote{
	This means, for any two vertices $s,t \in V$ in the same connected component, 
	we distinguish one $s$-$t$-path of minimum weight $d_G(s,t)$.
}
It is always possible to ensure this by
either randomly perturbing the edge weights with sufficiently small values, 
see~\cite{Mulmuley87IsolationLemma},
or by using a more complex deterministic method, known as \emph{lexicographic perturbation}~\cite{CaChEr13,Charnes52OptDegLP,HartvigsenMardon94AllPairsMinCut}.
The first method increases the preprocessing only by an additive $O(m)$ term,
but makes the preprocessing algorithm randomized.
Lexicographic perturbation, in turn, 
increases the time by an additive $\Otilde(mn)$~\cite{CaChEr13}.
In both cases, this conditioning step increases the preprocessing time in 
\Cref{thm:st-diameter-where-st-are-sets}
by at most a logarithmic factor.
(A more detailed description can be found at the beginning of \Cref{sec:f-dso-st}.)

We obtain several new $ST$-diameter oracles at the same time
by applying the reduction from \Cref{thm:st-diameter-where-st-are-sets}
to existing all-pairs $f$-DSOs.
We give an overview in \Cref{table:FDO-ST_from_all-pairs_DSO}.
The preprocessing times are omitted for brevity.
(The times for the underlying $f$-DSOs are reported in \Cref{table:all-pairs-DSOs}.)
We discuss the literature on distance sensitivity oracles extensively in \Cref{subsec:relatedWork_all-pairs}.

The previous result requires access to a distance sensitivity oracle
that can be queried with any pair of vertices.
In contrast, single-source $f$-DSO are usually easier to compute.
We show next how to use them in building $ST$-diameter oracles.
In general, all parameters of the reduction are improved compared to
\Cref{thm:st-diameter-where-st-are-sets}, except for the stretch.
The latter increases as a result of the lack of accuracy when the source of the underlying $f$-DSO is fixed.
The literature on single-source $f$-DSOs is discussed in \Cref{subsec:relatedWork_single-source}
and \Cref{table:ST-diam-from-single-source-DSO} lists the oracles obtained from previous works
via \Cref{thm:ST-diam-from-single-source-DSO}. 

\begin{table}[t]
\setstretch{1.2}
\centering
\caption{Properties of the fault-tolerant $ST$-diameter oracles for undirected graphs obtained via \Cref{thm:ST-diam-from-single-source-DSO} using single-source distance sensitivity oracles from the literature.
The parameters are the same as in \Cref{table:FDO-ST_from_all-pairs_DSO}.}
\label{table:ST-diam-from-single-source-DSO}
\begin{adjustbox}{max width=0.8\textwidth}
\begin{tabular}{cccccc}
\noalign{\hrule height 1pt}

\bf Sensitivity & \bf Stretch &\bf Space & \bf Query time  
& \bf References \\
\noalign{\hrule height 1pt}\\[-10pt]

$1$ & $7$  &  $\Otilde(n^{3/2})$ & $\Otilde(1)$ 
& \cite{BCFS21SingleSourceDSO_ESA,GuptaSingh18FaultTolerantExactDistanceOracle}\\[.5em]

$1$ & $7$  &  $\Otilde(M^{1/2} \nwspace n^{3/2})$ & $\Otilde(1)$ 
& \cite{BCFS21SingleSourceDSO_ESA}\\[.5em]

$1$ & $7+\varepsilon$  &  $\Otilde(n \nwspace \log(W)/\varepsilon)$ & $O(\log\log_{1+\varepsilon}(Wn))$ 
& \cite{BaswanaCHR20,BaswanaK13,Bilo16-esa}\\[.5em]

1 & 12 & $O(n)$ & $O(1)$ 
& \cite{Bilo16-esa}\vspace*{.5em}\\
\noalign{\hrule height 0.25pt}\\[-10pt]

$f \ge 1$  & $10f +7$  &  $\Otilde(fn)$ & $\Otilde(f^3)$  
& \cite{BiloG0P22Algorithmica}
\end{tabular}
\end{adjustbox}
\end{table}

\begin{restatable}{theorem}{stdiamfromsinglesourcedso}
\label{thm:ST-diam-from-single-source-DSO}
	Let $G=(V,E)$ be an undirected positively edge-weighted graph with $n$ vertices and $m$ edges.
	Let $S,T \subseteq V$ be two non-empty sets. 
	Assume access to a single-source $f$-DSO for $G$ with  stretch $\sigma$,
	space $\mathtt{S}$, query time $\mathtt{Q}$, and preprocessing time $\mathtt{P}$.
	There exists an $f$-FDO-$ST$ for $G$ with stretch $2+5\sigma$,
	space $O(\mathtt{S}+n)$, query time $O(f \,(\log(f)+\mathtt{Q}))$, and preprocessing time $\Otilde(\mathtt{P}+ m)$.
\end{restatable}

We now turn to the special case of the $ST$-diameter
with a single source or target, that is, for $|S| = 1$ or $|T| = 1$, respectively.
For the sake of clarity, 
whenever $S=\{s\}$ is a singleton, we use the term ``$sT$-diameter'' instead of ``$ST$-diameter'' 
or ``$\{s\}T$-diameter''; similarly for $T \,{=}\, \{t\}$.
In this setting, the underlying single-source $f$-DSOs really shine in that they allow to bring the stretch
back down again, along with all other parameters of the reduction.
\Cref{thm:sT-diam} is phrased for the $sT$-diameter but holds verbatim also for the $St$-diameter case.
\Cref{table:sT-diam} shows the resulting oracles. 

\begin{restatable}{theorem}{stdiam}
	\label{thm:sT-diam}
	Let $G=(V,E)$ be an undirected positively edge-weighted graph with $n$ vertices and $m$ edges.
	Let $s \in V$ be a vertex and $T \subseteq V$ a non-empty set.
	Assume access to a single-source $f$-DSO for $G$ with stretch $\sigma$,
	space $\mathtt{S}$, query time $\mathtt{Q}$, and preprocessing time $\mathtt{P}$.
	There exists an $f$-FDO-$sT$ for $G$ with stretch $1+2\sigma$,
	space $\mathtt{S}+O(n)$, query time $O(f  \,(\log(f) + \mathtt{Q}))$,
	and preprocessing time $\mathtt{P}+ \Otilde(m)$.
\end{restatable}

\begin{table}[tt]
\setstretch{1.2}
\centering
\caption{Properties of the fault-tolerant $sT$-diameter oracle (resp.\ $St$-diameter oracle) for undirected graphs obtained via \Cref{thm:sT-diam} using single-source distance sensitivity oracles from the literature.
The parameters are the same as in \Cref{table:FDO-ST_from_all-pairs_DSO}.}
\label{table:sT-diam}
\begin{adjustbox}{max width=0.8\textwidth}
\begin{tabular}{cccccc}
\noalign{\hrule height 1pt}

\bf Sensitivity & \bf Stretch &\bf Space & \bf Query time  
& \bf References \\
\noalign{\hrule height 1pt}\\[-10pt]

$1$ & $3$  &  $\Otilde(n^{3/2})$ & $\Otilde(1)$ 
& \cite{BCFS21SingleSourceDSO_ESA,GuptaSingh18FaultTolerantExactDistanceOracle}\\[.5em]

$1$ & $3$  &  $\Otilde(M^{1/2} \nwspace n^{3/2})$ & $\Otilde(1)$ 
& \cite{BCFS21SingleSourceDSO_ESA}\\[.5em]

$1$ & $3+\varepsilon$  &  $\Otilde(n \nwspace \log(W)/\varepsilon)$ & $O(\log\log_{1+\varepsilon}(Wn))$ 
& \cite{BaswanaCHR20,BaswanaK13,Bilo16-esa}\\[.5em]

1 & 5 & \zell{$O(n)$} & $O(1)$ 
& \cite{Bilo16-esa}\vspace*{.5em}\\
\noalign{\hrule height 0.25pt}\\[-10pt]

$f \ge 1$  & $4f +3$  &  $\Otilde(fn)$ & $\Otilde(f^3)$  
& \cite{BiloG0P22Algorithmica}
\end{tabular}
\end{adjustbox}
\end{table}

The other special case we discuss is that of the regular diameter, that is, $S = T = V$.
We provide two theorems showing how both all-pairs and single-source $f$-DSOs can be used to construct $f$-FDOs.
Note that both results hold for directed graphs, provided that the underlying distance sensitivity oracle
supports those.
The corresponding oracles can be found in \Cref{table:diameter_oracle,table:two_single_source_fdo}.

\begin{restatable}{theorem}{diameteroracle}
\label{thm:diameter_oracle}
Let $G$ be an (undirected or directed) positively edge-weighted graph with $n$ vertices and $m$ edges. Assume access to an $f$-DSO for $G$ with stretch $\sigma$, space $\mathtt{S}$, query time $\mathtt{Q}$, and  preprocessing time $\mathtt{P}$.
There exists an $f$-FDO for $G$ with stretch $1 + \sigma$, space $\mathtt{S}+O(1)$, query time $O(f^2 \nwspace \mathtt{Q})$,
and preprocessing time $\mathtt{P} + \Otilde(mn)$.
\end{restatable}

\begin{restatable}{theorem}{twosinglesourcefdo}
\label{thm:two_single_source_fdo}
Let $G$ be an (undirected or directed) positively edge-weighted graph with $n$ vertices and $m$ edges.
Assume access to a single-source $f$-DSO for $G$ with stretch $\sigma$, space $\mathtt{S}$, query time $\mathtt{Q}$, 
and preprocessing time $\mathtt{P}$.
There exists an $f$-FDO for $G$ with stretch $2 + 2\sigma$, space $O(\mathtt{S})$, query time $O(f \nwspace \mathtt{Q})$,
and preprocessing time $\Otilde(\mathtt{P} + m)$.
\end{restatable}

Finally, we also prove an information-theoretic lower bound on the space requirement
of approximate diameter oracles that support $f \ge 2$ edge failures.
Note that the bound is given in bits, while usually we measure the space in the number of $O(\log n)$-bit machine words.

\begin{table}[t]
\setstretch{1.2}
\centering
\caption{Properties of the fault-tolerant diameter oracles obtained via \Cref{thm:diameter_oracle} using all-pairs distance sensitivity oracles from the literature.
The applicable graph class (un-/directed, un-/weighted) is determined by the $f$-DSO.
The parameters are the same as in \Cref{table:FDO-ST_from_all-pairs_DSO}.}
\label{table:diameter_oracle}
\begin{adjustbox}{max width=\textwidth}
\begin{tabular}{cccccc}
\noalign{\hrule height 1pt}

\bf Sensitivity & \bf \zell{Stretch} &\bf Space & \bf Query time & \bf Ref. \\
\noalign{\hrule height 1pt}\\[-10pt]

1 & 2 & $\Otilde(n^2)$ & $O(1)$ & \cite{BeKa08,BeKa09}\\[.5em]

1 & 2& $\Otilde(n^2)$ & $O(1)$ & \cite{KarczmarzSankowski23SensitivityandDynamicDOs}\\[.5em]
 
1 & $(2k{-}1)(1{+}\eps)+1$ & $\Otilde (k^5 n^{1+1/k}/\eps^{4})$ & $O(k)$ & \cite{BaswanaK13}\vspace*{.5em}\\
\noalign{\hrule height 0.25pt}\\[-10pt]

2 & 2& $\Otilde(n^2)$ & $\Otilde(1)$ & \cite{DuanP09a}\vspace*{.5em}\\
\noalign{\hrule height 0.25pt}\\[-10pt]

$f = O(1)$ & $4 + \varepsilon$ & $\Otilde(n^{2-\frac{\alpha}{2(f+1)}} (\log(n)/\varepsilon)^{O(f)})$ & $O(n^{\frac{\alpha}{2}}/\varepsilon^2)$  & \cite{BCCCFKS24TheoretiCS} \vspace*{.5em}\\
\noalign{\hrule height 0.25pt}\\[-10pt]
  
$f = o(\frac{\log n}{\log \log n})$ & 
2 & 
$\Otilde(n^{3-\alpha})$ & 
$\Otilde(f^2n^{2-(1-\alpha)/f})$ & \cite{WY13}\\[.5em]
  
$f = o(\frac{\log n}{\log \log n})$ & 
$2+\varepsilon$ & 
$O(f n^{2+o(1)} \log(W)/\varepsilon^{f})$ & 
$\Otilde(f^7 \log \log W)$ & \cite{ChCoFiKa17}\vspace*{.5em}\\
\noalign{\hrule height 0.25pt}\\[-10pt]

$f \ge 1$ & 
2 & 
$O(fn^4)$ & 
$f^{O(f)}$ & \cite{DuRe22}\\[.5em]

$f  \ge 1 $ & 
2 & 
$O(n^{2+\alpha} M)$ & 
$\Otilde(f^4 Mn^{2-\alpha} + f^{\omega+2} Mn)$ & \cite{BrSa19}\\[.5em]

$f \ge 1$ & $(8k{-}2)(f{+}1)+1$ & $O\left(fkn^{1+1/k}\log{(Wn)}\right)$ & $\Otilde(f^3)$ & \cite{CLPR10}
\end{tabular}
\end{adjustbox}
\end{table}

\begin{table}[t]
\setstretch{1.2}
\centering
\caption{Properties of the fault-tolerant diameter oracles obtained via \Cref{thm:two_single_source_fdo} using single-source distance sensitivity oracles from the literature.
The applicable graph class (un-/directed, un-/weighted) is determined by the single-source $f$-DSO.
The parameters are the same as in \Cref{table:FDO-ST_from_all-pairs_DSO}.}
\label{table:two_single_source_fdo}
\begin{adjustbox}{max width=0.8\textwidth}
\begin{tabular}{cccccc}
\noalign{\hrule height 1pt}

\bf Sensitivity & \bf Stretch &\bf Space & \bf Query time  
& \bf References \\
\noalign{\hrule height 1pt}\\[-10pt]

1 & 4  &  $\Otilde(n^{3/2})$ & $\Otilde(1)$ 
& \cite{BCFS21SingleSourceDSO_ESA,GuptaSingh18FaultTolerantExactDistanceOracle}\\[.5em]

1 & 4  &  $\Otilde(M^{1/2} \nwspace n^{3/2})$ & $\Otilde(1)$ 
& \cite{BCFS21SingleSourceDSO_ESA}\\[.5em]

1 & $4+\varepsilon$  &  $\Otilde(n \log(W)/\varepsilon)$ & $O(\log\log_{1+\varepsilon}(Wn))$ 
& \cite{BaswanaCHR20,BaswanaK13,Bilo16-esa}\\[.5em]

1 & 6 & $O(n)$ & $O(1)$ 
& \cite{Bilo16-esa}\vspace*{.5em}\\
\noalign{\hrule height 0.25pt}\\[-10pt]

$f \ge 1$  & $4f +4$  &  $\Otilde(fn)$ & $\Otilde(f^3)$ 
& \cite{BiloG0P22Algorithmica}
\end{tabular}
\end{adjustbox}
\end{table}

\vspace*{-1em}
\begin{restatable}{theorem}{lowerbound}
\label{thm:lower-bound}
	Any $f$-FDO or $f$-FDO-ST for $n$-vertex graphs with sensitivity $f \ge 2$
	and a stretch of $\frac{5}{3}-\varepsilon$ for any $\varepsilon = \varepsilon(n) > 0$ 
	requires $\Omega(n^{3/2})$ bits of space.
\end{restatable}

\vspace*{-1em}
\subparagraph*{Outline.}
The remainder of this work is structured as follows.
The next three subsections review the literature on diameter oracles as well as all-pairs and single-source distance sensitivity oracles. 
We fix our notation in \Cref{sec:prelimiaries}. 
\Cref{sec:f-dso-st} presents our constructions of $f$-FDO-$ST$
for general sets $S, T \subseteq V$. 
In \Cref{sec:single-source-f-FDO,sec:FDOs}, we consider the respective 
special cases of a single source
and of the unrestricted diameter.
\Cref{sec:lower-bound} proves the space lower bound.

\subsection{Related Work on Fault-Tolerant Diameter Oracles}
\label{subsec:relatedWork-fdos}

Fault-tolerant diameter oracles were introduced by Henzinger, Lincoln, Neumann, and Vassilevska Williams~\cite{HenzingerL0W17}.
They showed that for a single failure in unweighted directed graphs, 
one can compute in time $\Otilde(mn + n^{1.5}\sqrt{Dm/\varepsilon})$,
where $\varepsilon\in (0,1]$ and $D$ is the diameter of the graph,
a $1$-FDO with $1+\varepsilon$ stretch that has $O(m)$ space, constant query time.
Bilò, Cohen, Friedrich, and Schirneck~\cite{BCFS21DiameterOracle_MFCS} improved the preprocessing time to $\Otilde(mn + n^2/\varepsilon)$, which is near-optimal assuming the combinatorial Boolean matrix multiplication conjecture 
(for details, see~\cite{HenzingerL0W17}).
Using fast matrix multiplication instead, their preprocessing time for dense graphs decreases 
to $\Otilde(n^{2.5794} + n^2/\varepsilon)$.

Bilò, Choudhary, Cohen, Friedrich, and Schirneck~\cite{Bilo22Extremal} addressed the problem of constructing $1$-FDOs with $o(m)$ space. 
They showed that for unweighted directed graphs with sufficiently large diameter $D=\omega(n^{5/6})$,
there is a \mbox{$1$-FDO} taking $\Otilde(n)$ space, with $1+\frac{n^{5/6}}{D} = 1 + o(1)$ stretch, and $O(1)$ query time.
It has a preprocessing time of $O(mn)$. 
In the same work, it was also shown that for any $\varepsilon > 0$ and graphs with diameter $D=\omega((n^{4/3} \log n)/(\varepsilon\sqrt m ))$, there is a $(1+\varepsilon)$-stretch $1$-FDO, 
with preprocessing time $O(mn)$, space $o(m)$, and constant query time.

For \emph{undirected} graphs the space requirement can be reduced.
There is a folklore construction that combines the $1$-DSO by Bernstein and Karger~\cite{BeKa09}
with the observation that in undirected graphs the eccentricity of an arbitrary vertex is a $2$-approximation of the diameter.
This results in an 1-FDO with stretch $2$ and constant query time
that takes only $O(n)$ space, see \cite{BCFS21DiameterOracle_MFCS,HenzingerL0W17}.

For $f \,{>}\, 1$ edge failures in undirected graphs with non-negative edge weights, Bilò~et~al.~\cite{BCFS21DiameterOracle_MFCS} also presented an $f$-FDO with stretch $f\,{+}\,2$, an $O(f^2 \log^2 n)$ query time, $\Otilde(f n)$ space, and $\Otilde(fm)$ preprocessing time.

We are not aware of any $O(n)$-sized, constant-stretch FDOs for \emph{directed} graphs with arbitrary diameter
in the literature prior to this work, not even for sensitivity $f=1$.
Also, no non-trivial $f$-FDOs with $o(f)$-stretch were known. 
To the best of our knowledge, we are the first to study the problem of general $f$-FDO-$ST$s with $S,T \neq V$.

We now discuss the known space lower bounds for FDOs.
Bilò, Cohen, Friedrich, and Schirneck~\cite{BCFS21DiameterOracle_MFCS} proved that $f$-FDOs with finite stretch must take $\Omega(fn)$ bits of space,
which nearly matches their construction (see above).
They also gave an $\Omega(m)$-bit bound for 1-FDOs with stretch $\sigma \,{<}\, 3/2$ for undirected \emph{unweighted} graphs,
or \emph{edge-weighted} graphs when $\sigma \,{<}\, 2$.
In a follow-up work, Bilò, Choudhary, Cohen, Friedrich, and Schirneck~\cite{Bilo22Extremal} generalized this to directed graphs. In particular, they showed that for directed unweighted graphs with diameter $D=O(n/\sqrt{m})$, any 1-FDO with stretch better than $\big(\frac{3}{2}-\frac{1}{D}\big)$ requires $\Omega(m)$ bits of space. They further proved that $f$-FDOs for digraphs require $\Omega(2^{f/2}n)$ bits of space, as long as $2^{f/2}=O(n)$.

\subsection{All-Pairs Distance Sensitivity Oracles}
\label{subsec:relatedWork_all-pairs}

\begin{table}[t]
\setstretch{1.2}
\centering
\caption{Existing all-pairs distance sensitive oracles.
The reported features are for undirected graphs, including for those $f$-DSOs that also support directed graphs.
The maximum edge weight for graphs with arbitrary positive weights is denoted $W$,
for integer-weighted graphs it is $M$.
The parameter $k \ge 1$ is a positive integer, $\varepsilon > 0$ a positive real, 
and $\alpha \in [0,1]$ is a real number in the unit interval.
We use $c > 1$ for a constant, and $\omega < 2.371339$ for the matrix multiplication exponent.}
\label{table:all-pairs-DSOs}
\begin{adjustbox}{max width=\textwidth}
\begin{tabular}{cccccc}
\noalign{\hrule height 1pt}

\bf Sensitivity & \bf \zell{Stretch} &\bf Space & \bf Query time & \bf\zell{Preprocessing \\ Time } & \bf Ref. \\
\noalign{\hrule height 1pt}\\[-10pt]

1 & 1 & $\Otilde(n^2)$ & $O(1)$ & $\Otilde(m n)$ & \cite{BeKa08,BeKa09}\\[.5em]

1 & 1 & $\Otilde(n^2)$ & $O(1)$ & $O(n^{2.529})$ & \cite{KarczmarzSankowski23SensitivityandDynamicDOs}\\[.5em]
 
1 & $(2k{-}1)(1{+}\eps)$ & $\Otilde (k^5 n^{1+1/k}/\eps^{4})$ & $O(k)$ & $O(kmn^{1+1/k})$ & \cite{BaswanaK13}\vspace*{.5em}\\
\noalign{\hrule height 0.25pt}\\[-10pt]

2 & 1& $\Otilde(n^2)$ & $\Otilde(1)$ & $\poly(n)$ & \cite{DuanP09a}\vspace*{.5em}\\
\noalign{\hrule height 0.25pt}\\[-10pt]

$f=O(1)$ & $3+\varepsilon$ & $\Otilde(n^{2-\frac{\alpha}{2(f+1)}} (\log(n)/\varepsilon)^{O(f)})$ & $O(n^{\frac{\alpha}{2}}/\varepsilon^2)$ &
$\Otilde(mn^{2-\frac{\alpha}{2(f+1)}} (\log(n)/\varepsilon)^{O(f)})$ & \cite{BCCCFKS24TheoretiCS}\vspace*{.5em}\\
\noalign{\hrule height 0.25pt}\\[-10pt]
  
$f = o(\frac{\log n}{\log \log n})$ & 
1 & 
$\Otilde(n^{3-\alpha})$ & 
$\Otilde(n^{2-(1-\alpha)/f})$ &
$O(Mn^{\omega+1-\alpha})$  & \cite{WY13}\\[.5em]  
  
$f = o(\frac{\log n}{\log \log n})$ & 
$1+\varepsilon$ & 
$O(f n^{2+o(1)} \log(W)/\varepsilon^{f})$ & 
$\Otilde(f^5 \log \log W)$ &
$O(fn^{5+o(1)} \log(W)/\varepsilon^{f})$  & \cite{ChCoFiKa17}\vspace*{.5em}\\
\noalign{\hrule height 0.25pt}\\[-10pt]

$f \ge 1$ & 
1 & 
$O(fn^4)$ & 
$f^{O(f)}$ &
$\Omega(n^f)$ &
\cite{DuRe22}\\[.5em]

$f \ge 1$ &  
1 & 
$O(f^4n^2\log(Mn))$ &  
$O((c f \log(Mn))^{O(f^2)})$ &  
$\Omega(n^f)$ &  
\cite{DeyGupta24}\\[.5em]


$f  \ge 1 $ & 
1 & 
$O(Mn^{2+\alpha})$ & 
$\Otilde(f^2 Mn^{2-\alpha} + f^{\omega} Mn)$ &
$\Otilde(Mn^{\omega + (3-\omega)\alpha})$ &
\cite{BrSa19}\\[.5em]

$f \ge 1$ & $(8k{-}2)(f{+}1)$ & $O\left(fkn^{1+1/k}\log{(Wn)}\right)$ & $\Otilde(f)$ &
$\poly(n)$  & \cite{CLPR10}


\end{tabular}
\end{adjustbox}
\end{table}

We now discuss previous works on distance sensitivity oracles.
They are summarized in \Cref{table:all-pairs-DSOs}.
For simplicity, the features for undirected graphs are reported, even if the oracle also works on digraphs.
This serves as the basis for \Cref{table:FDO-ST_from_all-pairs_DSO,table:diameter_oracle}.

The first distance-sensitive oracle was given by Demetrescu and Thorup for directed graphs \cite{DT02}.
It maintains exact distances for a single edge failure ($\sigma = 1$ and $f = 1$).
The space requirement of the oracle is $O(n^2 \log n)$ and its query time is $O(\log n)$.
This was later generalized to also handle a single vertex failures by Demetrescu, Thorup, Chowdhury, and Ramachandran~\cite{DeThChRa08}.
They presented an exact $1$-DSO
of size $O(n^2 \log n)$, with $O(1)$ query time and $\Otilde(mn^2)$ preprocessing time.
In two consecutive papers, Bernstein and Karger improved the preprocessing time first
to $O(n^2 \sqrt m)$ in~\cite{BeKa08} and then to $\Otilde(mn)$ in~\cite{BeKa09} 
(keeping the space and query time unchanged).
Baswana and Khanna \cite{BaswanaK13}
considered approximate $1$-DSOs for unweighted graphs.
They devised a data structure of size $O(k^5 n^{1+1/k} \log^3 (n)/\eps^4)$, with stretch $(2k{-}1)(1{+}\eps)$, and a $O(k)$ query time
for any positive integer parameter $k$.
Duan and Pettie~\cite{DuanP09a} considered the case of two failures (vertices and edges) with exact distances.
The size of their oracle is $O(n^2 \log^3 n)$, the query time is $O(\log n)$. 

Relying on fast matrix multiplication, the $1$-DSO by Chechik and Cohen \cite{ChCo20}
has a subcubic $\widetilde{O}(Mn^{2.873})$ preprocessing time for integer edge weights in the range $[-M,M]$ 
(that is, including negative values).
Their oracle has an $\widetilde{O}(1)$ query time.
This was later improved by Ren~\cite{Ren22Improved} as well as by Gu and Ren~\cite{GuRen21},
who obtained an $\Otilde(Mn^{2.5794})$ preprocessing time and constant query time.
For unweighted directed graphs, Karczmarz and Sankowski~\cite{KarczmarzSankowski23SensitivityandDynamicDOs} showed that one can construct in time $O(n^{2+\rho}) = O(n^{2.529})$ a $1$-DSO with $O(1)$ query time and $O(n^2)$ space.
Here, $\rho \approx 0.529$ is the solution of the equation $\omega(1,\rho,1) = 1 + 2\rho$,
where $\omega(a,b,c)$ is the smallest exponent such that one can multiply $n^a \times n^b$ and $n^b \times n^c$ matrices in $O(n^{\omega(a,b,c)})$ time. The data structure can also handle single vertex failures. 
The preprocessing time matches the best time known for computing APSP \cite{Zwick02DirectedAPSP}.

Regarding $f$-DSOs for larger values of $f$, Duan and Ren~\cite{DuRe22} gave an exact data structure for undirected weighted graphs with $O(fn^4)$ space, $f^{O(f)}$ query time.
However, their construction has an $\Omega(n^f)$ preprocessing time.
Later, Dey and Gupta~\cite{DeyGupta24} provided an alternative solution for undirected graphs with positive integer weights in $[1,M]$ with the same preprocessing time, but a better space of $O(f^4 n^2 \log(Mn))$ and a query time $O(c^{(f+1)^2} f^{8(f+1)^2}  \log^{2(f+1)^2}(Mn))$ for some constant $c > 1$.
The preprocessing times in \cite{DuRe22,DeyGupta24} are only polynomial for constant values of $f$.

Weimann and Yuster \cite{WY13} presented
an distance sensitivity oracle that can handle up to $f = o(\log{n}/\log\log{n})$ edge or vertex failures with 
an $\Otilde(n^{2-(1-\alpha)/f})$ query time and $O(Mn^{\omega+1-\alpha})$ preprocessing time for directed
graphs with integer weights in the range $[-M,M]$.
Here, $\alpha \in [0,1]$ is a real number and $\omega< 2.371339$ the matrix multiplication exponent~\cite{Alman25MoreAssymetryFasterMatrixMultiplication,VWilliams24MatrixMultiplicationAlphaToOmega,Duan23FasterMMAsymmetricHashing}.
Grandoni and Vassilevska Williams~\cite{GrandoniVWilliamsFasterRPandDSO_journal} constructed a 1-DSO
with subcubic $\Otilde(Mn^{\omega+\frac{1}{2}}+ Mn^{\omega+\alpha(4-\omega)})$
preprocessing time and sublinear $\Otilde(n^{1-\alpha})$ query time.
Inspired by the related topic of dynamic graphs,
van den Brand and Saranurak~\cite{BrSa19} gave a distance sensitive oracle
that can handle $f \le \log n$ (batch) updates, where an update is an edge insertion or deletion.
It has $\Otilde(Mn^{\omega + (3-\omega)\alpha})$ preprocessing time,
$\Otilde(f^2 Mn^{2-\alpha} + f^\omega Mn)$ update time,
and an $\Otilde(f Mn^{2-\alpha} + f^2 Mn)$ query time. 
The work of Karczmarz and Sankowski~\cite{KarczmarzSankowski23SensitivityandDynamicDOs} also contains a data structure for 
unweighted directed graphs with $\Otilde(n^\omega)$ preprocessing time and $O(n^2)$ space,
such that for any set $F$ of $f$ edge or vertex failures, the data structure can be \emph{updated} in time $O(f^{\omega-1} n)$ 
to then support distance queries with failures $F$ in $O(fn)$ time.

%
We now discuss approximate $f$-DSOs.
For undirected graphs with real edge weights bounded by $W$, 
Chechik, Langberg, Peleg, and Roditty~\cite{CLPR10} presented an $f$-DSO with stretch $(8k{-}2)(f{+}1)$, 
where again $k\geq 1$ is some integer parameter.
Notably the stretch scales with the number $f$ of supported edge failures.
The oracle takes space $O(fkn^{1+1/k}\log(Wn))$ and has an $\Otilde(f \log\log (Wn))$ query time. 
Later, Chechik, Cohen, Fiat and Kaplan~\cite{ChCoFiKa17} gave a solution whose stretch is independent of the sensitivity.
Namely, for every approximation parameter $1 > \varepsilon > \sfrac{1}{n}$ (including non-constant ones),
they designed an $(1 {+} \varepsilon)$-approximate $f$-DSO with a query time of $O(f^5 \log(n) \log\log(W))$,
space $O\big(f n^2  \log(W) \cdot (c\log (n)/\varepsilon)^f \big)$,
and $O\big(f n^5 \log(W) \cdot (c\log (n)/\varepsilon)^f \big)$ preprocessing time.
In the common case of a sensitivity $f = o(\log{n}/\log\log{n})$, polynomial weights $W = \poly(n)$, and constant $1 > \varepsilon >0$, 
their $f$-DSO has $1+\varepsilon$ stretch, takes $n^{2+o(1)}$ space, and has $\Otilde(1)$ query time, and $n^{5+o(1)}$ preprocessing time.


Recently, Bilò, Chechik, Choudhary, Cohen, Friedrich, Krogmann, and Schirneck~\cite{BCCCFKS24TheoretiCS} designed the first $f$-DSO with subquadratic space, constant stretch, and truly sublinear query time for undirected unweighted graphs with unique shortest paths. More precisely, for constants $f\ge 2$ and $\sfrac{1}{2} > \alpha > 0$, and for every $\varepsilon > 0$,
their data structure has stretch $3{+}\eps$, a space of
$\Otilde\big(n^{2-\frac{\alpha}{f+1}} \,{\cdot}\, (c\log n/\varepsilon)^{f+2} \big)$, query time of
$O(n^{\alpha}/\varepsilon^2)$, and preprocessing time
$\Otilde\big(mn^{2-\frac{\alpha}{f+1}} \,{\cdot}\, (c\log n/\varepsilon)^{f+1} \big)$. 
In a follow-up work, Bilò, Chechik, Choudhary, Cohen, Friedrich, and Schirneck~\cite{BCCCFSFOCS24} improved the approximation ratio
to (essentially) $1+\varepsilon$ multiplicative by introducing an \emph{additive} stretch of $2$.
They kept the space subquadratic and the query sublinear
by designing an $f$-DSO that, for any positive integer $\ell$, real number $(\ell{+}1)/2 \ge \alpha > 0$,
sensitivity $f \,{=}\, o(\log(n)/\log\log n)$, and $\varepsilon = \omega(\sqrt{\log n}/n^{\frac{\alpha}{2(\ell+1)(f+1)}})$,
has stretch $((1{+}\frac{1}{\ell})(1{+}\varepsilon), \nwspace 2)$,
space $n^{2- \frac{\alpha}{(\ell+1)(f+1)} + o(1)}/\varepsilon^{f+2}$, query time $O(n^\alpha/\varepsilon^2)$, 
and preprocessing time $n^{2+\alpha + o(1)} + mn^{2- \frac{\alpha}{(\ell+1)(f+1)} + o(1)}/\varepsilon^{f+1}$.

\subsection{Single-Source Distance Sensitivity Oracles}
\label{subsec:relatedWork_single-source}

\begin{table}[t]
\setstretch{1.2}
\centering
\caption{Existing single-source distance sensitive oracles.
The reported features are for undirected graphs, including for those $f$-DSOs that also support directed graphs.
The parameters are the same as in \Cref{table:all-pairs-DSOs}.}
\label{table:single-source-DSOs}
\begin{adjustbox}{max width=\textwidth}
\begin{tabular}{cccccc}
\noalign{\hrule height 1pt}

\bf Sensitivity & \bf Stretch &\bf Space & \bf Query time  & \bf\zell{Preprocessing \\ Time } & \bf Ref. \\
\noalign{\hrule height 1pt}\\[-10pt]

1 & 1  &  $\Otilde(n^{3/2})$ & $\Otilde(1)$ & $\Otilde(mn^{1/2}+n^2)$ & \cite{BCFS21SingleSourceDSO_ESA,GuptaSingh18FaultTolerantExactDistanceOracle}\\[.5em]

1 & 1  &  $\Otilde(M^{1/2} \nwspace n^{3/2})$ & $\Otilde(1)$ & $\Otilde(Mn^\omega)$ & \cite{BCFS21SingleSourceDSO_ESA}\\[.5em]

1 & $1+\varepsilon$  &  $\Otilde(n\log(W)/\varepsilon)$ & $O(\log\log_{1+\varepsilon}(Wn))$ & $\poly(n)$ & \cite{BaswanaCHR20,BaswanaK13,Bilo16-esa}\\[.5em]

1 & 2 & $O(n)$ & $O(1)$& $\Otilde(mn)$ &\cite{Bilo16-esa}\vspace*{.25em}\\
\noalign{\hrule height 0.25pt}\\[-10pt]

$f \ge 1$  & $2f +1$  &  $\Otilde(fn)$ & $\Otilde(f^2)$  & $\Otilde(fm) $ & \cite{BiloG0P22Algorithmica}
\end{tabular}
\end{adjustbox}
\end{table}

There are much fewer works on single-source DSOs.
They are summarized in \Cref{table:single-source-DSOs},
which in turn underpins \Cref{table:ST-diam-from-single-source-DSO,table:sT-diam,table:two_single_source_fdo}.
We first discuss undirected graphs.
Baswana and Khanna~\cite{BaswanaK13} showed that unweighted graphs
can be preprocessed in $O(m\sqrt{n/\varepsilon})$ time 
to compute a $(1{+}\varepsilon)$-stretch single-source \mbox{1-DSO} for edge and vertex failures.
The oracle has size $O(n\log n + n/\varepsilon^3)$ and constant query time.
For weighted graphs, they showed how an $O(n\log n)$ size oracle  
can report $3$-approximate distances for a single failure in $O(1)$ time. 
Bil{\`{o}}, Gualà, Leucci, and Proietti~\cite{Bilo16-esa} proved,
for a single edge failure in weighted graphs, one can compute an $O(n)$-size
oracle with stretch $2$, maintaining a constant query time. 
Also, a construction was provided that has $1+\varepsilon$ stretch,
with  $O(n\log(1/\varepsilon)/\varepsilon)$ space
and $O(\log(n)\log(1/\varepsilon)/\varepsilon)$ query time.

All the results stated until now were for a single edge or vertex failure.
In a more recent work, Bil{\`{o}}, Gual{\`{a}}, Leucci, and Proietti~\cite{BiloG0P22Algorithmica} gave a construction for multiple failures
that takes space $O(fn \log^2 n)$, and is computable in $\Otilde(fm)$ time.
The oracle reports $(2f{+}1)$-stretched distances in time $O(f^2 \log^2 n)$.

Bilò, Cohen, Friedrich, and Schirneck~\cite{BCFS21SingleSourceDSO_ESA} presented several additional single-source DSOs. For undirected unweighted graphs, they presented an oracle with $O(n^{3/2})$ space, query time $\Otilde(1)$,
and a $\Otilde(m\sqrt{n}+n^2)$ preprocessing time. 
For the case that the graph has integer edge weights in the range $[1,M]$ and one is willing to use fast matrix multiplication, 
they presented an alternative construction with $O(M^{1/2}n^{3/2})$ space, query time $\Otilde(1)$, 
and a $\Otilde(Mn^\omega)$ preprocessing time. 
Finally, for sufficiently sparse graphs with $m=O(M^{3/7}n^{7/4})$ edges they proved that a subquadratic-in-$n$ preprocessing time is possible.
Namely, they devised a single-source DSO with the same $O(M^{1/2}n^{3/2})$ size and $\Otilde(1)$ query time, 
but a preprocessing time of $\Otilde(M^{7/8}m^{1/2}n^{11/8})$.

For directed graphs, Baswana, Choudhary, Hussain, and Roditty~\cite{BaswanaCHR20} showed that one can preprocess graphs with edge weights in the range $[1,W]$ into an oracle with space $\Otilde (n \log(W)/\varepsilon)$
that reports $ (1{+}\varepsilon)$-approximate distances for a single edge/vertex failure in $\Otilde(\log\log_{1+\varepsilon}(Wn))$ time.
Gupta and Singh~\cite{GuptaSingh18FaultTolerantExactDistanceOracle} designed an exact 1-DSO for directed unweighted graphs.
The data structure has size $\Otilde(n^{3/2})$ and $\Otilde(1)$ query time.

\section{Preliminaries}
\label{sec:prelimiaries}

For a given graph $G=(V,E)$, possibly with positive edge weights, 
we denote by $d_G(u,v)$ the distance in $G$ from vertex $u \in V$ to vertex $v \in V$. 
Given two non-empty subsets $S,T \subseteq V$, the $ST$-\emph{diameter} of $G$ is defined as $\diam(G,S,T)=\max_{s \in S, t \in T} d_G(s,t)$. 
With a slight abuse of notation, when $S=\{s\}$, we also use $\diam(G,s,T)$ as a shorthand of $\diam(G,\{s\},T)$ for the $sT$-diameter. 
Likewise, we let $\diam(G,S,t)$ stand for $\diam(G,S,\{t\})$ in case $T = \{t\}$ is a singleton.
If $S=T=V$, we use $\diam(G)$ instead of $\diam(G,V,V)$.
The value $\diam(G,s,V) = \max_{v \in V} d_G(s,v)$ is the \emph{eccentricity} of $s$ in $G$.

For a set $F \subseteq E$ of edges, we denote by $G-F$ the graph obtained from $G$ 
by removing all the edges in $F$. 
If $H$ is a subgraph of $G$, we use $V(H)$ and $E(H)$ 
for the vertices and edges of $H$, respectively.
An $f$-\emph{edge fault-tolerant distance sensitivity oracle} ($f$-DSO) 
is a data structure that receives queries $(u,v,F)$ 
with $u,v \in V$ and $F \subseteq E$ with $|F| \le f$. 
It returns an estimate $\widehat{d}_{G-F}(u,v)$ of the $u$-$v$-distance in $G-F$ 
such that $d_{G-F}(u,v) \leq \widehat{d}_{G-F}(u,v) \leq \sigma \cdot d_{G-F}(u,v)$,
where $\sigma \ge 1$ is the \emph{stretch}.
The true quantity $d_{G-F}(u,v)$ is sometimes called the \emph{replacement distance} from $u$ to $v$.

An $f$-\emph{edge fault-tolerant \mbox{$ST$-diameter} oracle} (\mbox{$f$-FDO-$ST$}) with stretch $\sigma$
returns, upon query $F \subseteq E$ with $|F| \le f$,
an estimate $\widehat{D}$ of the $ST$-diameter of $G-F$
such that $\diam(G{-}F,S,T) \leq \widehat{D} \leq \sigma \cdot \diam(G{-}F,S,T)$.
If $S = \{s\}$ or $T = \{t\}$ are singletons or $S=T=V$ are both the whole vertex set,
we abbreviate the respective oracles as $f$-FDO-$sT$, $f$-FDO-$St$, and $f$-FDO.
We sometimes use the term \emph{fault-tolerant diameter} for $\diam(G{-}F)$ if the set $F$
is clear from context.

\section{\textit{ST}-Diameter Oracles}
\label{sec:f-dso-st}

Arguably, the technically most interesting construction is the reduction
from fault-tolerant $ST$-diameter oracles to distance sensitivity oracles
(\Cref{thm:st-diameter-where-st-are-sets} restated below).
This and the proof of \Cref{thm:sT-diam} 
in \Cref{sec:single-source-f-FDO} contain many techniques
that are also used to establish the other results later in \Cref{sec:FDOs}.

\stdiameterwherestaresets*

In the following, we assume that the shortest paths in $G$ are made unique.
We can thus identify a shortest path with its endpoints.
This allows us to save some preprocessing time and 
also makes the resulting data structure more space-efficient.
In fact, it will turn out to be the key ingredient to keep the space overhead
over the underlying $f$-DSO  subquadratic (in $n$).
However, the precise way how to make the paths unique influences the nature of the preprocessing.
As hinted in \Cref{sec:intro},
one can ensure a unique shortest path in a random fashion by slightly perturbing the edge weights.
Alternatively, lexicographic perturbation~\cite{CaChEr13,Charnes52OptDegLP,HartvigsenMardon94AllPairsMinCut}
provides a deterministic procedure
but adds an $\Otilde(mn)$ term to the running time.

We elaborate on the lexicographic perturbation approach. 
Assume that each of the $m$ edges is assigned a distinct integer label from the range $[1,m]$. 
For a path $P$, define the string $S_P \in \{0,1\}^m$
as the characteristic vector of its edge set $E(P) \subseteq E$.
That means, we have $S_P(i) = 1$ iff the edge with label $i$ belongs to $P$. 
Given two different paths $P_1$ and $P_2$ of equal length between the same endpoints,
we break the tie in favor of $P_1$ if the string $S_{P_1}$ is lexicographically smaller than $S_{P_2}$, that is, if at the first index $i$ where $S_{P_1}(i) \neq S_{P_2}(i)$, 
we have $S_{P_1}(i) < S_{P_2}(i)$. 

This ordering can be implemented efficiently when running Dijkstra's algorithm
to compute a shortest-path tree from every vertex. 
Specifically, we maintain the current shortest-path tree in a top-tree data structure~\cite{tarjan2005self} augmented with a lowest common ancestor (LCA) structure, supporting path queries. 
Consider a $\text{relax}(u,v)$ operation where $d_G(s,u) + w(u,v) = d_G(s,v)$. 
In this case, we break ties by comparing the candidate path obtained by extending the path from $s$ to $u$ with the edge $(u,v)$, and the current path to $v$. 
Let $z = \mathrm{LCA}(u, v)$. 
Using the top-tree structure, we can lexicographically compare these two paths by identifying the maximum-labeled edge along each path segment beyond $z$, and selecting the path with the smaller such maximum-label, and update the top-trees and LCA data structures accordingly. 
A description of the lexicographic perturbation technique with full details
can be found for example in~\cite[Section 6.2]{CaChEr13}. 
In the following, whenever we refer to \emph{the} shortest path between vertices $s$ and $t$,
we mean the respective path in the tree rooted at the source $s$.

Let us first give a coarse outline of the remaining section.
We define two sets $S_F$ and $T_F$ depending on $S$, $T$ and the failure sets $F$.
We then show how they can be used to estimate the fault-tolerant $ST$-diameter.
The main part of the section consists of constructing an auxiliary data structure
that efficiently computes $S_F$ and $T_F$.
We then show how to reduce the space of that structure, 
albeit at the expense of a higher query time.

Fix a set $F \subseteq E$ of at most $f$ edges
and recall that we use $V(F)$ to denote the set of endpoints of edges in $F$.
We distinguish two subsets of $V(F)$, $S_F$ and $T_F$.
Intuitively, a vertex $v$ belongs to $S_F$ if it is on some shortest $ST$-path
and the portion of that path from the source to $v$ survives even after the failures;
likewise, $v$ is in $T_F$ if the portion from it to the target is not affected by failing edges.
For the exact definition, let $\pi_{x,y}$ denote the unique shortest path from vertex $x$ to $y$ in $G$.
Some $v \in V(F)$ is in $S_F$ if there exist $s \in S$ and $t \in T$ such that
$v$ is on the shortest path $\pi_{s,t}$ and the subpath from $s$ to $v$ contains no failing edge.
The definition of $T_F$ is analogous with $t$ in place of $s$.

We use two crucial properties of the sets $S_F$ and $T_F$ and their elements.
First, let $v \in S_F$ be such an element as witnessed by some $s \in S$ and $t \in T$.
The prefix of $\pi_{s,t}$ from $s$ to $v$ has length\footnote{
	Since the shortest paths are unique,
	the prefix of $\pi_{s,t}$ from $s$ to $v$ \emph{is} the shortest path $\pi_{s,v}$,
	but this stronger property is not needed here.
} $d_G(s,v)$ and,
since $v \in S_F$, this prefix does not contain any edge of $F$.
This gives $d_{G-F}(s,v) = d_G(s,v)$.
Similarly, we have $d_{G-F}(v,t) = d_G(v,t)$ for $v \in T_F$ and their respective targets $t$.
Secondly, $S_F$ and $T_F$ are subsets of $V(F)$ and \emph{not} of $S$ and $T$.
This will help to reduce the query times tremendously. 
The cardinalities $|S_F|, |T_F|$ are in $O(f)$ while $S$ and $T$ may be much larger.
The core of this section is to prove that the information in $S_F,T_F$ is enough
to estimate the fault-tolerant $ST$-diameter.

\subsection{Query Algorithm}
\label{subsec:ST-diam_query}

Before we describe how $S_F$ and $T_F$ are computed and stored, we present the query algorithm of our $ST$-diameter oracle.
Let $\D$ denote the $f$-DSO with stretch $\sigma$
that is assumed in \Cref{thm:st-diameter-where-st-are-sets}.
Let $F$ be the query set and suppose $S_F$ and $T_F$ to be given.
For every pair of vertices $(u,v) \in S_F \times T_F$,
the diameter oracle queries $\D$ with the triple $(u,v,F)$ to obtain a $\sigma$-approximation of $d_{G-F}(u,v)$.
The $f$-FDO-$ST$ returns the value
\begin{equation*}
	\widehat{D} = \diam(G,S,T)+\max_{(u,v) \in S_F\times T_F}\D(u,v,F).
\end{equation*}

\noindent
The value $\diam(G,S,T)$ can be precomputed.
So given $S_F$ and $T_F$, the time needed to obtain $\widehat{D}$ is $O(f^2\mathtt{Q})$,
where $\mathtt{Q}$ is the query time of the $f$-DSO $\D$.

We next prove the stretch of our oracle.
Using the sets $S_F$ and $T_F$ in place of $S$ and $T$ have the advantage that they make the query time independent of the sizes $|S|$ and $|T|$.
However, this bears the danger that the oracle may overestimate the \mbox{$ST$-diameter} of $G{-}F$.
After all, for arbitrary vertices $u \notin S$ and $v \notin T$, the distance $d_{G-F}(u,v)$  may not be bounded by any multiple of $\diam(G{-}F,S,T)$.
The crucial part in the proof of the next lemma is to show that in the specific case $u \in S_F$ and $v \in T_F$,
$d_{G-F}(u,v)$ is at most three times the fault-tolerant $ST$-diameter.

\begin{lemma}
\label{lemma:st-fdo}
	The $f$-FDO-$ST$ has stretch $1+3\sigma$.
\end{lemma}

\begin{proof}[Proof.]
	Let $s \in S$ and $t \in T$ be two vertices. 
	We first show $d_{G-F}(s,t) \leq \widehat{D}$.
	That means,
	the returned value never underestimates the $ST$-diameter of $G{-}F$.
	We only need to prove the case in which some of the failing edges in $F$ belong to $\pi_{s,t}$;
	otherwise, $d_{G-F}(s,t) = d_{G}(s,t) \leq \diam(G,S,T) \leq \widehat{D}$.
	Consider all failing edges on $\pi_{s,t}$ and let $u^*$ be the endpoint of such an edge 
	that is closest to $s$.
	Also, let $v^*$ be the endpoint closest to $t$ among all edges in $F \cap E(\pi_{s,t})$. 
	We thus have $u^* \in S_F$ and $v^* \in T_F$
	and $d_{G-F}(s,u^*)=d_G(s,u^*)$ as well as $d_{G-F}(v^*,t)=d_G(v^*,t)$. 
	Since $\pi_{s,u^*}$ and $\pi_{v^*,t}$ are vertex-disjoint subpaths of $\pi_{s,t}$, we get
	\begin{equation*}
		d_{G-F}(s,u^*)+d_{G-F}(v^*,t) = d_G(s,u^*)+d_G(v^*,t) \le d_G(s,t) \leq \diam(G,S,T).
	\end{equation*}
	
	Recall that the distance sensitivity oracle $\D$ satisfies $d_{G-F}(u,v) \le \D(u,v,F)$ for all admissible queries $(u,v,F)$.
	Since $(u^*, v^*) \in S_F \times T_F$, we have 
	\begin{equation*}
		d_{G-F}(u^*,v^*) \le \max_{(u,v) \in S_F\times T_F} d_{G-F}(u,v) \le \max_{(u,v) \in S_F\times T_F}\D(u,v,F).
	\end{equation*}	
	Combining the two facts gives
	\begin{align*}
		d_{G-F}(s,t) &\leq d_{G-F}(s,u^*)+d_{G-F}(u^*,v^*)+d_{G-F}(v^*,t)\\
			&\leq \diam(G,S,T)+\max_{(u,v) \in S_F\times T_F}\D(u,v,F) =  \widehat{D}.
	\end{align*}

	We now prove the other inequality $\widehat{D} \leq (1{+}3\sigma) \cdot \diam(G{-}F,S,T)$,
	where $\sigma$ is the stretch of $\D$. 
	Let $u \in S_F$ and $v \in T_F$ be two vertices.	
	Our claim will follow once we have established that $d_{G-F}(u,v) \le 3 \cdot \diam(G{-}F,S,T)$.
	By definition, the reason why $u$ is in $S_F$ is that
	there exist $s_u \in S$ and $t_u \in T$ such that $u$ is
	on the shortest $s_u$-$t_u$-path $\pi_{s_u,t_u}$ in $G$
	and $d_{G-F}(s_u,u) = d_{G}(s_u,u)$.
	We get
	\begin{equation*}
		d_{G-F}(s_u,u) = d_{G}(s_u,u) \le d_G(s_u,t_u) \le \diam(G,S,T).
	\end{equation*}	
	By the same argument, the fact $v \in T_F$ is witnessed by a vertex $t_v \in T$
	with $d_{G-F}(v,t_v) \le \diam(G,S,T)$.
	Now note that the graph $G{-}F$ is undirected, which means $d_{G-F}(s_u,u) = d_{G-F}(u,s_u)$.
	The $ST$-diameter of $G{-}F$ is also never smaller than that of $G$.
	Thus,
	\begin{align*}
    	d_{G-F}(u,v) &\leq d_{G-F}(u,s_u)+d_{G-F}(s_u,t_v)+d_{G-F}(t_v,v)\\
			&\leq  2 \cdot \diam(G,S,T)+\diam(G{-}F,S,T)  \leq 3 \cdot \diam(G{-}F,S,T).
	\end{align*}
	
	The final bound on the stretch follows from the fact that the $f$-DSO overestimates any fault-tolerant distance
	by at most a factor $\sigma$,
	\begin{align*}
		\widehat{D} &= \diam(G,S,T)+\max_{(u,v) \in S_F\times T_F}  \D(u,v,F)\\
			&\le \diam(G,S,T)+\max_{(u,v) \in S_F\times T_F}  \sigma \cdot d_{G-F}(u,v)\\
			&\le \diam(G,S,T) + \sigma \cdot 3 \diam(G{-}F,S,T) \leq (1+3\sigma) \diam(G{-}F,S,T).	\end{align*}
\end{proof}

\subsection{Basic Data Structure for \texorpdfstring{\boldmath{$S_F$}}{S'} 
	and \texorpdfstring{\boldmath{$T_F$}}{T'}}
\label{subsec:ST-diameter_datastructure}

Given the failure set $F$, 
the set $S_F$ contains all $v \in V(F)$ such that there exists $s \in S$ and $t \in T$
for which $v$ is the closest endpoint to $s$ of a failing edge on the shortest path $\pi_{s,t}$.
The set $T_F$ is defined analogously.
We now describe the data structure that computes $S_F$ and $T_F$.
The exposition focuses on $S_F$, the structure for  $T_F$ uses the same ideas.

The first variant of our construction takes space $O(n^2)$
in addition to the space $\mathsf{S}$ of the underlying $f$-DSO.
For any vertex $v \in V$, let $\widetilde{T}_v$ denote the shortest path tree
of the graph $G$ rooted in $v$.
We use the notation $\widetilde{T}_v$ to avoid confusion with the sets $T$ and $T_F$.
The preprocessing algorithm of our data structure computes, for each $v \in V$,
the tree $\widetilde{T}_v$ and marks certain nodes in it.
A node $s \in S$ is marked in the tree $\widetilde{T}_v$ if there is some $t \in T$ 
such that $v$ lies on the shortest path $\pi_{s,t}$.
For any two vertices $s \in S$ and $t \in T$, 
$\pi_{s,t}$ contains $v$ if and only if $d_G(s,t) = d_G(s,v)+ d_G(v,t)$.
We used here that the shortest paths are unique.
The algorithm can thus compute all-pairs distances in $G$ in time\footnote{
	The time needed for this step reduces to $O(mn)$ in case $G$ is unweighted,
	has only small integer weights, or floating points in exponent-mantissa representation
	using Thorup's algorithm~\cite{Thorup99UndirectedSSSPinLinearTime}.
}
$O(mn + n^2 \log n)$
and mark the relevant nodes in \emph{all} trees with the obvious $O(n|S||T|)$-time algorithm.

\begin{figure}[t]
	\centering
	\includegraphics[scale=1]{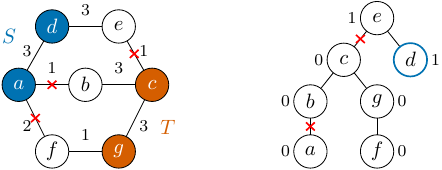}
	\caption{Illustration of the basic data structure for the set $S_F$.
		On the left is the example graph with failing edges $\{a,b\}$, $\{a,f\}$, and $\{c,e\}$.
		On the right is the shortest-path tree $\widetilde{T}_e$ rooted in $e$.
		The nodes are annotated with their $count$.
		Only the node $d \in S$ is marked as none of the shortest paths that start in $a \in S$ 
		and end in some vertex of $T$ use $e$.
		The set $F_0$ consists of the edges $\{a,b\}$ and $\{c,e\}$,
		the former is dominated by the latter.	
		Since the marked node $d$ remains reachable from $e$ in $\widetilde{T}_e - F_0$,
		the root $e$ is included in the set $S_F$.
	}
\label{fig:marked_tree}
\end{figure}

Additionally, each node $u$ of $\widetilde{T}_v$ is annotated with its pre- and post-order traversal number and the value $count_v(u)$,
the latter being the number of marked nodes in the subtree $(\widetilde{T}_v)_u$ rooted in $u$.
Each node except for the root $v$ also stores a pointer to its parent.
The annotations do not increase the preprocessing time or storage space by more than a constant factor.
In particular, for a fixed tree $\widetilde{T}_v$, all values $count_v(u)$ are computable in $O(n)$ time in bottom-up fashion.
The total time needed for preprocessing 
is thus $O(mn + n^2 \log n + n|S||T|)$ and the space is $O(n^2)$.
 
To answer a query $F$, the algorithm scans all the vertices $v \in V(F)$
and decides which of them to include in $S_F$.
An illustration is given in \Cref{fig:marked_tree}.
The graph $\widetilde{T}_v-F$ is a forest of rooted trees.
Possibly some of the trees degenerated into isolated vertices.
We claim that $v \in S_F$ holds if and only if $\widetilde{T}_v-F$
contains a marked node that is still reachable from $v$.
To see this, observe that the uniqueness of shortest paths implies their consistency.
That means, if $v$ is on some shortest path $\pi_{s,t}$,
then the subpath from $s$ to $v$ is the shortest path $\pi_{s,v}$
(and the one from $v$ to $t$ is $\pi_{v,t}$).
Since $G$ is undirected, $\pi_{s,v}$ is contained in $\widetilde{T}_v$.
Now consider the two defining conditions for $v \in S_F$.
Vertex $v$ lying on the path $\pi_{s,t}$ means $s \in S$ is marked,
and $\pi_{s,v}$ not having a failing edge is equivalent to $s$ being reachable in $\widetilde{T}_v-F$.

To check whether $v \in S_F$, the algorithm first computes
the set $F_0  = F \cap E(\widetilde{T}_v)$ of those edges whose failure
affect a shortest path starting in $v$.
An edge $\{u,w\} \in F$ is in $F_0$ iff the parent of $u$ in $\widetilde{T}_v$ is $w$
or the parent of $w$ is $u$.
When talking about tree edges,
we adopt the notation $(u,w) \in F_0$ to indicate that $u$ is the parent of $w$.
The edges are still meant to be undirected.
Next, we define a notion of domination among edges in $F_0$.
We say that an edge $(x,y) \in F_0$ 
is \emph{dominated} by another edge $(u,w) \in F_0$ 
if $(x,y)$ is in the subtree $(\widetilde{T}_v)_w$.
We show how to identify all dominated edges in time $O(f \log f)$
using the tree-traversal numbers stored during preprocessing.
For a node $u$, let $pre(u)$ denote the pre-order number of $u$
and $post(u)$ its post-order number.
We further say $x$ is \emph{a parent node} if there exists a $y$ such that $(x,y) \in F_0$,
that is, $x$ is \emph{the} parent node in a failing edge in $\widetilde{T}_v$.
Likewise, we say $w$ is \emph{a child node} if there is some edge $(u,w) \in F_0$.
These notions are not mutually exclusive.
An edge in $F_0$ is dominated iff its parent node $x$ has some child node as an ancestor
(which may be $x$ itself).
Finally, some child node $w$ is an ancestor of $x$ if and only if $pre(w) \le pre(x)$ and $post(w) \ge post(x)$.

To exploit this connection, we first sort the full set $V(F_0)$
by descending pre-order number, resulting in the list $\textit{pre-all}^{\downarrow}$.
Let further $\textit{post-child}^{\downarrow}$ denote only the child nodes 
ordered by descending post-order number.\footnote{
	The list $\textit{post-child}^{\downarrow}$ only containing child nodes
	ensures that we find the remaining child node with maximum post-order number
	in constant time.
}
We scan through $\textit{pre-all}^{\downarrow}$ consuming both lists in the process.
Let $x$ be the first entry of $\textit{pre-all}^{\downarrow}$.
If it is not a parent node (i.e., only a child node and not both), 
then remove $x$ from both lists.
If $x$ is a parent node, check the current-first entry $w$ of $\textit{post-child}^{\downarrow}$.
Since $w$ has not been removed yet,
it is a child with $pre(w) \le pre(x)$.
Therefore, the edge $(x,y) \in F_0$ is dominated if and only if $post(w) \ge post(x)$.
After this check, $x$ is again removed from both lists (if present) and we continue
with the new first entry of $\textit{pre-all}^{\downarrow}$.
The bottleneck is to sort the two lists in time 
$O(|V(F_0)| \cdot \log |V(F_0)|) = O(f \log f)$.
We denote the set of the remaining non-dominated edges by $F_0'$.

Recall that we include the vertex $v$ in $S_F$
if the number of marked node
in the connected component of $v$ in $\widetilde{T}_v-F$ is positive.
Clearly, this connected component is the same in $\widetilde{T}_v -F_0$
and even in $\widetilde{T}_v -F_0'$, i.e., when ignoring dominated edges.
We now use the stored $count$s.
The value $count_v(w)$ is the number of marked nodes in the subtree $(\widetilde{T}_v)_w$
\emph{before} failures. 
No two edges in $F_0'$ dominate each other.
Hence, the subtrees of $\widetilde{T}_v$ rooted at the child nodes of edges in $F_0'$
span exactly the vertices that are \emph{not} in the same connected component as $v$
in $\widetilde{T}_v -F_0'$.
By summing the $count$-values over all child nodes\footnote{
	The reason to compute the subset $F_0'$ of non-dominated
	edges is to avoid double-counting here.
}
and comparing it to $count_v(v)$,
we get the desired information.

Since we have to construct the set $F_0'$ for each root $v \in V(F)$ separately,
computing $S_F$ takes total time $O(f^2 \log f)$, same for $T_F$.

\subsection{A More Compact Data Structure}
\label{subsec:ST-diameter_smallSensitivity}

We now design an alternative data structure that has a space overhead of 
$O(2^{f/2} \nwspace n^{3/2} \sqrt{\log n})$, rather than $O(n^2)$, 
provided that $f \le \log_2(n)$.
This will, however, increase the query time.
The idea is to store the information of the trees $\widetilde{T}_v$ in a more compact way. 
For every vertex $v \in V$, we define a new representation $\T_v$ of the tree $\widetilde{T}_v$ by first removing unnecessary parts and then replacing long paths with single edges.
This corresponds to the two steps of the compression described below.
For the first one, we need the following definition.
We say a subtree $\T_v$ of $\widetilde{T}_v$ preserves the \emph{source-to-leaf reachability}
if, for every set $F$ of up to $f$ edges,
there is a marked node of $\widetilde{T}_v$ that is reachable from the source $v$ in $\widetilde{T}_v-F$ 
if and only if there is a leaf of $\T_v$ that is reachable from $v$ in $\T_v-F$.

\paragraph*{First Compression Step}
Our goal is to preserve the source-to-leaf reachability
when compressing $\widetilde{T}_v$.
Let $S' \subseteq S$ be the set of nodes that are marked in $\widetilde{T}_v$.
Using an algorithm by Petruschka~\cite{Petruschka24FaultEquivalenLCA_ArXiv},
we compute in time $O(n + |S'|) = O(n)$ a set $\L_v \subseteq V$ of at most $2^{f-1}$ nodes 
such that, for every $F \subseteq E$ with $|F| \le f$,
a node of $S'$ is reachable from $v$ in $\widetilde{T}_v-F$
iff a node in $\L_v$ is still reachable.\footnote{
	The result in \cite{Petruschka24FaultEquivalenLCA_ArXiv} is stated in terms of vertex failures.
	We can adapt it by subdividing every edge in $\widetilde{T}_v$
	and applying it to a failure set of up to $f$ of those new ``edge-vertices''.
}
The main difference between $S'$ and $\L_v$ is that the size of the latter is bounded
by a function of $f$.
Of course, if $|S'| \le 2^{f-1}$, we can forgo this step entirely.
This is particularly relevant if already $S$ and $T$ have fewer than $2^{f-1}$ elements.
We define the tree representation $\T_v$ as the smallest subtree of $\widetilde{T}_v$
that contains the root $v$ and $\L_v$.
The notation $\L_v$ is mnemonic of ``leaves''.
It is clear from the construction of $\L_v$ that $\T_v$ preserves 
the root-to-leaf reachability of $\widetilde{T}_v$.

\paragraph*{Second Compression Step}
After the first compression step, the tree $\T_v$ contains at most $2^{f-1}$ leaves.
It may still be the case that the total number of nodes is large
due to the presence of very long paths connecting two branch nodes,
that is, nodes with two or more children.
We now contract those paths to single edges.

We use the critical path framework of Alon, Chechik, and Cohen~\cite{AlonChechikCohen19CombinatorialRP}.
Let $x$ and $y$ be two consecutive branch nodes of $\T_v$,
i.e., $x$ is an ancestor of $y$ in $\T_v$ and the internal nodes
of the path $\pi_{x,y}$ all have degree $2$ in $\T_v$. 
We say that $\pi_{x,y}$ is \emph{critical} if it has at least $L = \sqrt{n \log_2(n)/2^{f-1}}$ edges.
If so, we replace $\pi_{x,y}$ in $\T_v$ with the \emph{representative edge} $\{x,y\}$
and add $\pi_{x,y}$ to the set $\P$ of critical paths.
This collection can be computed in total time $O(n^2)$ by depth-first searches in each of the $n$ trees.
Eventually, it contains at most $q = O(n^2/L) = O(2^{f/2} \nwspace n^{3/2}/\sqrt{\log n})$ paths
as each tree contributes $O(n/L)$ paths. 
A deterministic hitting set $B$ for $\P$ can be computed in time $\Otilde(qL + n^2/L) = \Otilde(n^2)$.
It has cardinality
\begin{align*}
	|B| &= O\!\left(n \, \frac{\log q}{L}\right) 
		 = O\!\left(2^{f/2}\nwspace n^{1/2} \left(\sqrt{\log n} 
			+ \frac{f}{\sqrt{\log n}} \right)\!\right)\\[.25em]
		&= O\!\left(2^{f/2}\nwspace n^{1/2} \sqrt{\log n} \right).
\end{align*}
The last estimate uses the assumption $f \le \log_2(n)$.

For each \emph{pivot} $b \in B$, we store the shortest-path tree $\widetilde{T}_b$ of the original graph $G$ rooted in $b$. 
Note that a critical path $\pi_{x,y} \in \P$ is contained in $\widetilde{T}_b$ 
for any pivot $b$ that hits $\pi_{x,y}$.
Moreover, this $b$ is the lowest common ancestor of $x$ and $y$ in the tree $\widetilde{T}_b$. 

The final data structure consists of the trees $\{\T_v\}_{v \in V}$ and $\{\widetilde{T}_b\}_{b \in B}$
with some small amount of extra information.
For each tree ($\widetilde{T}_b$ and $\T_v$ alike),
we store a data structure that answers lowest common ancestor (LCA) queries in constant time.
Such a structure can be constructed in time/space that is linear in the size of the tree \cite{Beetal00}.
The trees $\T_v$ additionally have a dictionary for the set of their representative edges.
Since there are at most $O(2^f)$ branch nodes in any tree,
this also upper bounds the number of representative edges.
The dictionaries can be implemented to take space $O(2^f)$
and have a constant worst-case query time~\cite{HagerupMiltersenPagh01DeterministicDictionaries}.
For any representative edge $\{x,y\}$ of $\T_v$,
we additionally have a pointer to the tree $\widetilde{T}_b$ of an (arbitrary) pivot $b$ 
that hits $\pi_{x,y}$.
After the second compression step, any tree $\T_v$ contains at most 
$O(2^f L) = O(2^{f/2} \nwspace n^{1/2} \sqrt{\log n})$ nodes.
Since there are $n$ such trees, the LCA structures, and dictionaries take 
$O(2^{f/2} \nwspace n^{3/2} \sqrt{\log n})$ space.
Likewise, the trees $\widetilde{T}_b$ for all the pivots in $B$
take $O(|B|n) = O(2^{f/2} \nwspace n^{3/2} \sqrt{\log n})$ space.
If $M = \max(|S|,|T|)$ is at most $2^{f-1}$,
then, each tree $\T_v$ has $|\L_v| \le M$ leaves
and we can replace every occurrence of $2^{f-1}$ in the above derivation by $M$.

Regarding the preprocessing time,
running Petruschka's algorithm~\cite{Petruschka24FaultEquivalenLCA_ArXiv} for all $n$ trees $\widetilde{T}_v$ (if necessary)
and obtaining the pivots $B$ takes total time $\Otilde(n^2)$.
Computing the LCA structures for the 
trees $\{\T_v\}_{v \in V}$ and $\{\widetilde{T}_b\}_{b \in B}$
can be done in time that is linear in their total size,
i.e., $O(2^{f/2} \nwspace n^{3/2} \sqrt{\log n}) = \Otilde(n^2)$.
Therefore, the $\Otilde(mn + n |S||T|)$ time to construct and mark the original trees $\{\widetilde{T}_v\}_v$
(see \Cref{subsec:ST-diameter_datastructure}) still dominates the preprocessing.

Consider a set $F$ of at most $f$ failing edges.
We now explain how to compute $S_F$ from the compressed trees in time $O(2^f f^2)$
(resp., in time $O(f^2 M)$).
The set $S_F$ are those vertices $v \in V(F)$ 
such that there is a shortest $s$-$t$-path $\pi_{s,t}$
with $s \in S$ and $t \in T$ that contains $v$,
but the prefix $\pi_{s,v}$ is free from failing edges.
Since the root-to-leaf reachability is preserved, this is equivalent
to having a leaf in $\L_v$ that is reachable from $v$ in $\T_v-F$. 
Conceptually speaking, we can safely remove all non-representative edges in $\T_v$ that are in $F$.
For the representative edges $\{x,y\}$,
we have to look up information in the \emph{other} tree $\widetilde{T}_b$ 
with root $b \in B$ associated with $\{x,y\}$.
That edge is removed from $\T_v$ if there is a failing edge in $F$
that is contained in the path $\pi_{x,y}$ in $\widetilde{T}_b$.
We then return to $\T_v$ check whether some leaf from $\L_v$ is still connected to the root.

In more detail, we first compute the non-representative failing edges in $\T_v$.\footnote{
	In \Cref{subsec:ST-diameter_datastructure}, we used the parent pointers 
	of individual nodes for this task.
	The same solution is possible here and would increase the total space requirement only by a constant factor.
}
For some edge $\{u, w\} \in F$, we check whether it is in the dictionary of representative edges of $\T_v$.
If not, we test whether the LCA of $u$ and $w$ in $\T_v$ is in $\{u,w\}$.
If so, we have $\{u,w\} \in E(\T_v)$.
This can be done for all of $F$ with $O(f)$ constant-time look-ups and LCA queries.
To handle an representative edges $\{x,y\}$,
we follow the pointer to the tree $\widetilde{T}_b$.
The edge $\{u,w\}$ lies on $\pi_{x,y}$ in $\widetilde{T}_b$
if and only if one of two conditions hold:
\vspace*{.25em}
\begin{enumerate}
	\item[(\emph{i})] the LCA of $u$ and $x$ is $u$ and the LCA of $w$ and $x$ is $w$;
	\vspace*{.25em}
	\item[(\emph{ii})] the LCA of $u$ and $y$ is $u$ and the LCA of $w$ and $y$ is $w$.
\end{enumerate}
\vspace*{.25em}

The whole check takes time $O(2^f f)$ 
since there are at most $O(2^f)$ representative edges per tree.
Slightly abusing notation, let $F_0$ be the set of edges removed from $\T_v$,
including possibly some representative edges.
$F_0$ may no longer be a subset of $F$,
but we still have $|F_0| \le |F|$.
We test for each leaf in $\L_v$ and each edge in $F_0$ whether the latter
lies on the path from the root $v$ to that leaf.
This takes a constant number of LCA queries per test and thus $O(2^f f)$ time.
In summary, the tree $\T_v$ can be queried in time $O(2^f f)$ 
and we have to do this for every $v \in V(F)$,
resulting in a $O(2^f f^2)$ query time.

\section{Single-Source \textit{sT}-Diameter Oracles}
\label{sec:single-source-f-FDO}

The next question we address is how to construct $sT$-diameter oracles,
that is, a data structure that reports maximum replacement distance 
from a single source $s$ to a set of vertices $T$.
To build the $f$-FDO-$sT$, we use a single-source distance sensitivity oracle,
giving us access to estimates of all replacement distances from $s$.
We restate the relevant theorem below.
Its proof uses similar ideas as those shown in \Cref{sec:f-dso-st},
but the single-source setting allows for improvements in all parameters.

\stdiam*

\begin{proof}[Proof.]
	Let $\D$ denote the underlying single-source $f$-DSO.
	The preprocessing algorithm for the $f$-FDO-$sT$ constructs $\D$ with source $s$
	and computes the shortest-path tree $\widetilde{T}_s$ of $G$ rooted in $s$.
	Each node $v \in V$ is annotated
	with a pointer to its parent node in $\widetilde{T}_s$ and its respective number
	in the pre-order and post-order traversal of the tree.
	Similarly as above, the algorithm also computes the value $count(v)$ for every $v$.
	This is the number of all vertices in the set $T$ that are descendants of $v$ in $\widetilde{T}_s$,
	possibly including $v$ itself.
	(A difference to \Cref{subsec:ST-diameter_smallSensitivity} is that 
	no marking of vertices is needed here.)
	Finally, the preprocessing algorithm stores the maximum distance $\diam(G,s,T) =\max_{t\in T}d_G(s,t)$ 
	from the root to the vertices in the set $T$.
	The preprocessing takes total time $\mathtt{P}+ O(m + n \log n)$ in general weighted graphs
	and, again, can be reduced to $\mathtt{P} + O(m)$
	for certain classes of weights~\cite{Thorup99UndirectedSSSPinLinearTime}.
	Storing the oracle and the tree takes $\mathtt{S} + O(n)$ space.

	For the query, let $F \subseteq E$ be a set of up to $f$ failing edges
	and $F_0 = F \cap E(\widetilde{T}_s)$ those failures that are in the tree.
	Consider the forest of rooted trees $\widetilde{T}_s - F_0$ and
	let $R_F$ be the set of roots of those trees that contain some vertex from $T$.
	For any $v \in V$, let $\D(v,F)$ be the $\sigma$-approximation of 
	the replacement distance $d_{G-F}(s,v)$ computed by the DSO $\D$.
	Our $sT$-diameter oracle answers the query $F$ by reporting
	\begin{equation*}
		\widehat{D}= \diam(G,s,T) + \max_{r\in R_F} \D(r,F).
	\end{equation*}
	
	We defer the discussion of how to obtain the set $R_F$ and first prove 
	that $\widehat{D}$ is an \mbox{$(1{+}2\sigma)$-approximation} of $\diam(G{-}F,s,T)$.
	Let $t \in T$ be a target vertex
	and $r \in R_F$ the root of the subtree in $\widetilde{T}_s - F_0$ that contains $t$.
	By the choice of $r$, we have $d_{G-F}(r,t) = d_{G}(r,t)$.
	Also, it holds that $d_{G-F}(s,r) \le \D(r,F)$ since the oracle $\D$ never underestimates
	the replacement distance.
	Combining this with the definition of our estimate $\widehat{D}$ gives
	\begin{align*}
		d_{G-F}(s,t) &\le d_{G-F}(s,r) + d_{G-F}(r,t) = d_{G-F}(s,r) + d_G(r,t)\\
			&\le d_{G-F}(s,r) + d_G(s,t) \le \D(r,F) + \diam(G,s,T) \le \widehat{D}
	\end{align*}
	For the last inequality, we need the maximization over all roots in $R_F$
	in the definition of our estimate $\widehat{D}$.
	As the above argument holds for any target $t$, we get $\diam(G{-}F,s,T)\leq \widehat{D}$.

	For the other direction,
	let $r^* = \argmax_{r \in R_F} \D(r,F)$ be the root that maximizes the \emph{approximate}
	replacement distance from $s$ returned by $\D$.
	Let further $t\in T$ be some target vertex in the subtree of $\widetilde{T}_s - F_0$ 
	that is rooted in $r^*$.
	We thus have
	\begin{equation*}
		d_{G-F}(s,r^*) \leq d_{G-F}(s,t)+d_{G}(t,r^*) \leq  d_{G-F}(s,t)+ d_{G}(t, s) 
		\leq 2 \cdot d_{G-F}(s,t).
	\end{equation*}
	The first inequality uses that $G$ is undirected, so that we can go ``up'' 
	the tree from $t$ to $r^*$.
	Here it is important that we do not just use any root in $\widetilde{T}_s - F_0$,
	but one that actually has some target vertex in its tree.
	The last inequality stems from the replacement distance from $s$ to $t$ being never
	smaller than the original $s$-$t$-distance.
	
	We combine this with the fact that $\D(r^*,F)$ is a $\sigma$-approximation of the true
	replacement distance $d_{G-F}(s,r^*)$.
	\begin{multline*}
		\widehat{D}= \diam(G,s,T) + \D(r^*,F) \leq \diam(G,s,T)+\sigma  \cdot d_{G-F}(s,r^*)\\
			\leq \diam(G,s,T)+2\sigma \cdot  d_{G-F}(s,t) \leq (1+2\sigma) \cdot \diam(G{-}F,s,T).
	\end{multline*}
	In summary, we have $\diam(G{-}F,s,T)\leq \widehat{D} \le (1+2\sigma) \cdot \diam(G{-}F,s,T)$,
	as desired.

	The forest $\widetilde{T}_s- F_0$ consists of $|F_0|+1 \le f+1$ trees,
	so the set $R_F$ contains at most this many roots.
	Given $R_F$, computing $\widehat{D}$ takes time $O(f\mathtt{Q})$,
	where $\mathtt{Q}$ is the query time of the oracle $\D$.
	It remains to show how to compute $R_F$ from $F$ in time $O(f \log f)$.
	Recall that we know the parent of every non-root node in $\widetilde{T}_s$.
	We use it to first obtain $F_0$ from $F$ in time $O(f)$
	as an edge $\{a,b\}$ is in $\widetilde{T}_s$ if and only if $a$ is parent of $b$ or vice versa. 
	
	For any failing edge $e\in F_0$, let $r(e)$ be the endpoint of $e$ that is farther
	from the source $s$.
	Note that the roots of the trees in $\widetilde{T}_s - F_0$ are either $s$
	or such a ``lower'' endpoint of a failing edge.
	We use $R=\{r(e) \,{\mid}\, e\in F_0\} \cup \{s\}$ to denote them.
	The subset $R_F \subseteq R$ of the roots we are interested in 
	are those whose tree contain a vertex from the target set $T$.
	We have to decide for each $r \in R$, whether it is in $R_F$.
	
	We claim this can be done in time $O(f \log f)$.
	To do so, we first partition $R$ depending on the relative position
	of its elements in the tree $\widetilde{T}_s$ without failures.	
	Consider a different root $r' \,{\in}\, R{\setminus} \{r\}$.
	We say $r'$ is an \emph{immediate descendant} of $r$
	if $r'$ is a descendant of $r$ in $\widetilde{T}_s$
	and there is no other element of $R$ on the $r$-$r'$-path.
	Let $R(r)$ be the (possibly empty) subset of immediate descendants of $r$.
	Conversely, we say $r$ is the \emph{immediate ancestor} of all the $r' \in R(r)$.
	The source $s$ is the only root that has no immediate ancestor and
	the sets $\{R(r)\}_{r \in R}$ partition $R{\setminus}\{s\}$,
	see \Cref{fig:tree_reconstruction}.
	
	\begin{figure}[t]
	\centering
	\includegraphics[scale=1]{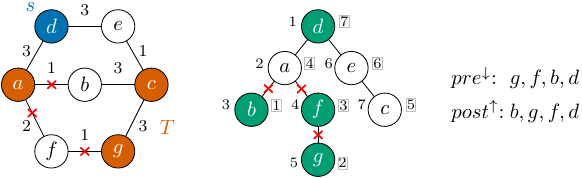}
	\caption{Illustration of the proof of \Cref{thm:sT-diam}.
		On the left is the example graph, now with a single source $d$, target set $T = \{a,c,g\}$,
		and failing edges $\{a,b\}$, $\{a,f\}$, and $\{f,g\}$.
		In the center is the shortest-path tree $\widetilde{T}_d$ rooted in $d$.
		Each node is annotated with its pre-order number (left) and post-order number (right, boxed).
		The root set $R = \{b,d,f,g\}$ is highlighted (in green).
		The subsets of immediate descendants are $R(d) = \{b,f\}$ and $R(f) = \{g\}$,
		the sets $R(b)$ and $R(g)$ are empty.
		On the right are the two sorted copies of $R$.
	}
\label{fig:tree_reconstruction}
\end{figure}
	
	We can compute this partition using the traversal numbers again,
	but this time with a somewhat simpler approach
	since we do not need to distinguish between parent and child nodes.
	Let $pre(r)$ denote the pre-order number of $r$ and $post(r)$ its post-order number.	
	Root $r$ is a (proper but not necessarily immediate) ancestor of $r'$
	if and only if $pre(r) < pre(r')$ and $post(r) > post(r')$.
	Out of all proper ancestors of $r'$, the \emph{immediate} ancestor $r$ 
	is the one that has minimum post-order number.
	We again use two sorted lists, both contain all roots from $R$.
	The first list $pre^{\downarrow}$ is sorted by decreasing pre-order number,
	and the second one $post^{\uparrow}$ is sorted by \emph{increasing} post-order number.
	Let $r'$ be the first entry of $pre^{\downarrow}$.
	Hence, the entries of $post^{\uparrow}$ that come after $r'$
	are exactly its proper ancestors.
	By the above argument, the immediate ancestor $r$ we are looking for
	is the successor of $r'$ in $post^{\uparrow}$.
	If we now remove $r'$ from both lists,
	we ensure that again all remaining roots have smaller pre-order number than the 
	(new) first entry of $pre^{\downarrow}$.
	We iterate this until both lists only contain the source $s$
	(which is without any proper ancestors),
	and in each step assign the current first entry of $pre^{\downarrow}$ to its immediate ancestor.
	This way, any $r \in R$ gets assigned all the roots in $R(r)$. 
	Sorting the two lists takes time $O(|R| \log |R|) = O(f \log f)$.

	Finally, we can compute the subset $R_F \subseteq R$.
	Observe that a vertex $r \in R$ lies in $R_F$
	if and only if there is at least one vertex of $T$ 
	that falls into the subtree of $\widetilde{T}_s$ rooted in $r$ 
	but not in any of the subtrees rooted in (proper) descendants of $r$ in $R$. 
	We decide this via the values $count(r)$, that is the number of targets from $T$
	in the subtree of $\widetilde{T}_s$ rooted in $r$. 
	We only need to consider the \emph{immediate} descendants in $R(r)$. 
	If the element of $T$ is in some lower subtree, 
	then it is also accounted for by a lower root. 	
	Putting this all together shows that some $r \in R$ is in $R_F$ 
	if and only if $count(r) > \sum_{r'\in R(r)}count(r')$.
	Since the sets $\{R(r)\}_r$ are a partition, those checks can be performed
	for all $r \in R$ in total time $O(f)$.
\end{proof}

We now generalize the result to multiple sources.
That means, we build an $f$-FDO-$ST$ for a general set $S$
while still only having access to a \emph{single-source} DSO.
The naive solution would be to iterate the above construction for the $sT$-case
for every source $s \in S$.
We indeed reduce the case of diameter estimation from multiple sources to that from a single source,
but do it more efficiently.
As it turns out, it is enough to construct the $sT$-diameter oracle 
for only \emph{two} arbitrary vertices $s \in S$ and $t \in T$.

\begin{restatable}{lemma}{stdiamfromstdiam}
\label{lem:ST-diam-from-sT-diam}
	Let $G=(V,E)$ be an undirected positively edge-weighted graph.
	Let \mbox{$S,T \subseteq V$} be non-empty sets of vertices, 
	and $s\in S$ and $t\in T$ be two vertices.
	Assume access to an $f$-FDO-${s}T$ and an \mbox{$f$-FDO-$tS$} for $G$
	with respective stretches $\sigma_{sT}$ and $\sigma_{tS}$,
	preprocessing times $\mathtt{P}_{sT}$ and $\mathtt{P}_{tS}$,
	space requirements $\mathtt{S}_{sT}$ and $\mathtt{S}_{tS}$,
	and query times $\mathtt{Q}_{sT}$ and $\mathtt{Q}_{tS}$.
	There exists an $f$-FDO-$ST$ for $G$
	with stretch $\sigma_{sT} + \sigma_{tS} + \min(\sigma_{sT},\sigma_{tS})$, 
	preprocessing time $\mathtt{P}_{sT} + \mathtt{P}_{tS}$,
	space $\mathtt{S}_{sT} + \mathtt{S}_{tS}$,
	and query time $\mathtt{Q}_{sT} + \mathtt{Q}_{tS} + O(1)$.
\end{restatable}

\begin{proof}[Proof.]
	Consider a failure set $F \,{\subseteq}\,E$.
	Let $\D_{sT}(F)$ be the \mbox{$\sigma_{sT}$-approximation} of $\diam(G-F,s,T)$,
	and $\D_{tS}(F)$ the $\sigma_{tS}$-approximation of $\diam(G-F,t,S)$
	obtained from the respective oracles.
	To answer the query $F$, the $f$-FDO-$ST$ outputs
	\begin{equation*}
		\widehat{D} = \D_{sT}(F)+ \D_{tS}(F) + \min\!\Big( \D_{sT}(F), \D_{tS}(F) \Big).
	\end{equation*}
	Consider an arbitrary pair of vertices $(s_0,t_0)\in S\times T$.
	We get
	\begin{align*}
		d_{G-F}(s_0,t_0) &\leq d_{G-F}(s_0,t)+d_{G-F}(t,s)+d_{G-F}(s,t_0)\\
				&\leq \D_{tS}(F) + \min\!\Big( \D_{sT}(F), \D_{tS}(F) \Big) + \D_{sT}(F) = \widehat{D}.
	\end{align*}
	Here, we use that the graph is undirected and therefore
	$d_{G-F}(s_0,t) \le \D_{tS}(F)$ and $d_{G-F}(t,s) \le \D_{sT}(F)$ both hold.
	So $\widehat{D}$ is never smaller than any distance between vertices from $S$ and $T$,
	it never underestimates $\diam(G{-}F,S,T)$.
	
	For the other part of the approximation,
	we insert the stretches of the two estimates $\D_{sT}(F)$ and $\D_{tS}(F)$.
	\begin{multline*}
		\widehat{D} \le \sigma_{sT} \cdot \diam(G{-}F,s,T)\, +  \sigma_{tS} \cdot \diam(G{-}F,t,S)\\
			\qquad\qquad\quad \min\!\Big(\sigma_{sT} \cdot \diam(G{-}F,t,S),\ \sigma_{tS} \cdot \diam(G{-}F,s,T)\Big).
	\end{multline*}
	Since both $\diam(G{-}F,s,T)$ and $\diam(G{-}F,t,S)$ are at most $\diam(G{-}F,S,T)$,
	this shows that $\widehat{D}$ estimates the latter with a stretch of 
	$\sigma_{sT} + \sigma_{tS} + \min(\sigma_{sT},\sigma_{tS})$.
\end{proof}

Combining \Cref{thm:sT-diam} and \Cref{lem:ST-diam-from-sT-diam} 
gives a reduction from the construction of $f$-FDO-$ST$ to that of single-source $f$-DSOs.
However, it results in a data structure with a stretch of $3+6\sigma$,
where $\sigma$ is the original stretch of the single-source $f$-DSO.
We improve this 
by not treating \Cref{lem:ST-diam-from-sT-diam} as a black box.

\stdiamfromsinglesourcedso*

\begin{proof}[Proof.]
	Let vertices $s \,{\in}\, S$ and $t \,{\in}\, T$ be arbitrary.
	The preprocessing algorithm of the \mbox{$f$-FDO-$ST$} uses 
	the single-source $f$-DSO twice, once for the source $s$ and once for the target $t$,
	to construct an $f$-FDO-$sT$ $\D_{sT}$ and an $f$-FDO-$tS$ $\D_{tS}$ both with stretch $1+ 2\sigma$.
	The details are given in \Cref{thm:sT-diam}.
	
	For a set $F \subseteq E$ of at most $f$ edge failures,
	let $\D_{sT}(F)$ and $\D_{tS}(F)$ be the respective $(1+2\sigma)$-approximations 
	of $\diam(G-F,s,T)$ and  $\diam(G-F,t,S)$.
	Further, let $\D_{st}(F)$ be a $\sigma$-approximation of $d_{G-F}(s,t)$,
	obtained from the DSO with the single source $s$.
	The query algorithm outputs the estimate
	\begin{equation*}
		\widehat{D}= \D_{tS}(F)+ \D_{st}(F) + \D_{sT}(F).
	\end{equation*}
	Let $(s_0,t_0)\in S\times T$.	
	By the same arguments as in the proof of \Cref{lem:ST-diam-from-sT-diam},
	we get
	\begin{align*}
		&d_{G-F}(s_0,t_0) \leq d_{G-F}(s_0,t)+d_{G-F}(t,s)+d_{G-F}(s,t_0)\\
			&\quad\quad\leq \D_{tS}(F)+\D_{st}(F)+\D_{sT}(F)\\
			&\quad\quad\leq (1{+}2\sigma) \cdot \diam(G{-}F,t,S) + \sigma \cdot d_{G-F}(t,s)
				+ (1{+}2\sigma) \cdot \diam(G{-}F,t,S).
	\end{align*}
	Again, $\diam(G{-}F,t,S)$, $ d_{G-F}(t,s)$, and $\diam(G{-}F,t,S)$
	are all bounded by $\diam(G{-}F,S,T)$.
	The sum of the stretch coefficients is $2 \cdot (1{+}2\sigma) + \sigma = 2 + 5\sigma$.
\end{proof}

\section{Unrestricted Diameter Oracles}
\label{sec:FDOs}

In this section, we discuss how to construct oracles that estimate the ordinary diameter,
the maximum distance between any two vertices, i.e., $S = T = V$.
As before, we first use all-pairs distance sensitivity oracles to build our data structures 
and afterwards describe the single-source case as an alternative.

\subsection{Diameter Oracles from All-Pairs DSOs}
\label{subsec:diameter_oracles_all-pairs}

There is a simple reduction from the construction of $f$-FDOs to that of \mbox{$f$-DSOs}
that results in a stretch of $1+\sigma$,
where $\sigma$ is the stretch of the distance sensitivity oracle.
The reduction can also handle directed graphs.

\diameteroracle*

\begin{proof}[Proof.]
	Let $\D$ denote the assumed $f$-DSO and, for any two vertices $x,y \in V(G)$
	and set $F \subseteq E(G)$ of at most $f$ failing edges,
	let $\D(u,v,F)$ be the reported \mbox{$\sigma$-approximation} 
	of the replacement distance $d_{G-F}(u,v)$.
	We precompute the original diameter $\diam(G)$ by running Dijkstra's algorithm
	from every vertex in time $O(mn + n^2 \log n)$.
	Our $f$-FDO stores $\D$ and the value $\diam(G)$.
	When queried with a set $F$, it returns
	\begin{equation*}
		\widehat D = \diam(G)+\max_{u,v \in V(F)} \D(u,v,F).
	\end{equation*}
	The query time is $|F|^2  \nwspace \mathtt{Q} = O(f^2 \nwspace \mathtt{Q})$.	
	
	To prove the stretch of the diameter oracle,
	we first show that $\widehat{D}$ is always at least the fault-tolerant diameter,
	namely, we claim  $d_{G-F}(s,t) \leq \widehat{D}$ for all $s,t \in V(G)$.
	Let $\pi$ be a shortest path from $s$ to $t$ in $G$.
	We only need to prove the lower bound when $\pi$
	contains some failing edges in $F$ as otherwise 
	$d_{G-F}(s,t)=d_G(s,t) \leq \diam(G) \leq \widehat D$.
	Let $x_s$ (resp., $x_t$) be the vertex of $V(F)$ that is closest to $s$ (resp., $t$) in $\pi$.
	By choice of $x_s$ and $x_t$ we have that $d_{G-F}(s,x_t)=d_G(s,x_s)$
	and $d_{G-F}(x_t,t)=d_G(x_t,t)$ since no failure occurs on the subpaths 
	$\pi[s,x_s]$ and $\pi[x_t,t]$. 
	Moreover, $d_G(s,x_s)+ d_G(x_t,t)\leq d_G(s,t) \le \diam(G)$.
	The triangle inequality now gives
	\begin{align*}
    	d_{G-F}(s,t) &\leq d_{G-F}(s,x_s)+d_{G-F}(x_s,x_t)+d_{G-F}(x_t,t)\\
     		&\leq d_G(s,x_s)+ \max_{u,v \in V(F)} \D(u,v,F) + d_{G}(x_t,t)\\
				& \leq \diam(G)+ \max_{u,v \in V(F)} \D(u,v,F) = \widehat D.
	\end{align*}
	
	For the upper bound on the stretch, recall that the DSO $\D$ has stretch $\sigma$.
	\begin{multline*}
    	\widehat D = \diam(G)+\max_{u,v \in V(F)} \D(u,v,F) 
    		\leq \diam(G{-}F)+\sigma \cdot \max_{u,v \in V(F)} d_{G-F}(u,v)\\
            	\leq \diam(G{-}F)+\sigma \cdot \diam(G{-}F) = (1 + \sigma) \cdot \diam(G-F).
	\end{multline*}
\end{proof}

\subsection{Diameter Oracles from Single-Source DSOs}
\label{subsec:diameter_oracle_singe-souce}

Recall that single-source distance sensitivity oracles only report the replacement distances from 
one vertex $s$ specified at preprocessing time.
We now use such data structures to built $f$-FDOs.
The following theorem is based on the intuition that in an undirected graph
twice the eccentricity of an arbitrary vertex is a $2$-approximation for the diameter.
We transfer this intuition to directed graphs under edge failures.

\twosinglesourcefdo*

\begin{proof}[Proof.]
	Suppose the graph $G = (V,E)$ is directed
	and let $\cev{G} = (V, \cev{E})$ denote the directed graph 
	in which the orientation of every edge is inverted compared to $G$.	
	Let $s\in V$ be any vertex. 
	We precompute two single-source distance sensitivity oracles $\D$ and $\cev{\D}$
	for the respective graphs $G$ and $\cev{G}$, both with source $s$.
	In other words, for any $v \in V$ and any set $F$ of up to $f$ edge failures,
	$\D(v,F)$ is a $\sigma$-approximation of $d_{G-F}(s,v)$;
	accordingly, $\cev{D}(v,F)$ approximates $d_{G-F}(v,s)$.
	(Of course, if $G$ is undirected, one single-source $f$-DSO suffices.)
	Also, we compute and store the respective eccentricities of $s$ in $G$ and $\cev{G}$,
	that is, $\diam(G,s,V)$ and $\diam(G,V,s)$.
	
	Let $F$ be a set of up to $f$ edge failures in $G$
	and $V(F)$ the set of its end points.
	Our oracles returns the value
	\begin{equation*}
		\widehat{D} = \diam(G,s,V)+\diam(G,V,s)+\max_{v\in V(F)} \D(v,F)+
 \max_{v\in V(F)} \cev{\D}(v,F).
	\end{equation*}
	As the query time of both $\D$ and $\cev{\D}$ is $\mathtt{Q}$,
	the query time of our oracle is $O(f \nwspace \mathtt{Q})$.

	Consider any two vertices $u,w \in V$.
	As before, we first show $d_{G-F}(u,w) \le \widehat{D}$.
	Recall that $\pi_{u,w}$ is the shortest path from $u$ to $w$.
	If $E(\pi_{u,w})\cap F=\emptyset$, then
	\begin{equation*}
		d_{G-F}(u,w)=d_{G}(u,w)\leq d_{G}(u,s)+d_{G}(s,w)\leq \diam(G,V,s)+ \diam(G,s,V)\leq\hat{D}.
	\end{equation*}
	Otherwise, let $x$ (resp., $y$) be the first (resp., last) vertex from $V(F)$ on $\pi_{u,w}$.
	Then the prefix between $u$ and $x$ as well as the suffix between $y$ and $v$
	are non-faulty edge-disjoint subpaths of $\pi_{u,w}$, 
	thus $d_{G-F}(u,x)+d_{G-F}(y,w)\leq d_{G}(u,w)$.
	Using this, we now obtain
	\begin{align*}
		d_{G-F}(u,w) & \leq d_{G-F}(u,x)+d_{G-F}(x,y)+d_{G-F}(y,w)\\
		 & \leq d_{G}(u,w)+d_{G-F}(x,y).
	\end{align*}
	We bound each of the last two terms as follows.
	\begin{align*}
		d_{G}(u,w) & \leq d_{G}(u,s)+d_{G}(s,w)\leq \diam(G,V,s)+ \diam(G,s,V),\\
		d_{G-F}(x,y) & \leq d_{G-F}(x,s)+d_{G-F}(y,s)
			\leq\max_{v\in V(F)}\mathcal{D}(v,F)+\max_{v\in V(F)}\overset{\leftarrow}{\mathcal{D}}(v,F).
	\end{align*}
	So, summing these inequalities, we get $d_{G-F}(u,w)\leq\hat{D}$ as required.
 \end{proof}

\section{Space Lower Bound}
\label{sec:lower-bound}

We now discuss the limits of compression for diameter oracles.
Our lower bound shows that, for sufficiently small stretch,
any such data structure must take a super-linear amount of space.
We prove this using an argument from information theory,
which results in the bound holding unconditionally.
In particular, it remains true even if we allow for an unbounded query time.

\lowerbound*

The stretch requirement stems from the fact that any such $f$-FDO-$ST$ can distinguish
whether the graph $G{-}F$ has fault-tolerant diameter 3 or 5.
The core of the proof of \Cref{thm:lower-bound}
is to establishing the fact that \emph{any} data structure 
with this ability requires $\Omega(n^{3/2})$ bits of space.

\begin{lemma}
\label{lem:lower-bound}
	For infinitely many positive integers $n$, there exists a graph $G = (V,E)$ with $n$ vertices 
	(and two sets $S,T \subseteq V$)
	such that any data structure that decides for all pairs of edges $e,e'\in E$,
	whether $G-\{e,e'\}$ has diameter (resp., $ST$-diameter) 
	at most 3 or at least 5 requires $\Omega(n^{3/2})$ bits of space.
\end{lemma}

\subsection{The graphs \emph{H} and \emph{G}}
\label{subsec:lower_bound_H}

For the remainder of this section, let $n=6N$ 
for some integer $N \ge 4$ that is a perfect square.
We first construct an auxiliary undirected graph $H$ with $n$ vertices.
Let indices $i,j$ range over the set $[\sqrt{N} \nwspace] = \{1, 2, \dots, \sqrt{N} \nwspace \}$
and let $k$ range over $\{0,1\}$.
Instead of the usual index notation like $b_{i,j,k}$, 
we use $b[ \nwspace i,j,k]$ to enhance readability.
We define four pairwise disjoint sets of vertices

\begin{multicols}{2}
	\begin{itemize}
		\item $A=\{a[\nwspace i,j]\}_{i,j \in [\sqrt{N}]}$\vspace*{.5em}
		\item $B= \{b[\nwspace i,j,k]\}_{i,j \in [\sqrt{N}],k \in \{0,1\}}$
		\item $C=\{c[\nwspace i,j,k]\}_{i,j \in [\sqrt{N}],k \in \{0,1\}}$\vspace*{.5em}
		\item $D = \{d[\nwspace i,j]\}_{i,j \in [\sqrt{N}]}$
	\end{itemize}
\end{multicols}

\noindent
The sets $A,B,C,D$ have respective cardinalities $N,2N,2N,$ and~$N$.
The vertex set of $H$ is $V(H)=A\cup B\cup C\cup D$.
The edges in $H$ are shown in \Cref{table:lb-5by3-stretch}.
They are defined using relations between the indices of the participating vertices.
For example, the second line states 
that an edge $\{ b[\nwspace i,j,k],\ b[\nwspace x,y,z] \}$ between elements of $B$ 
exists if and only if  \emph{either} $i = x$ holds \emph{or} $j = y$.
The indices $k, z \in \{0,1\}$ can be arbitrary.
\Cref{fig:lower_bound} illustrates the symmetries of $H$.
Next, we establish some features of the auxiliary graph.

\begin{table}[t]
\caption{
	The edges of the auxiliary graph $H$ in the proof of \Cref{lem:lower-bound}.
	The left column lists certain subsets of $V(H) \times V(H)$,
	the middle column shows the respective pairs of vertices with their indices,
	and the right column gives the condition under which those vertices indeed form an edge.
	The symbol $\oplus$ denotes the exclusive disjunction.
	The edges are intended to be undirected.}
\label{table:lb-5by3-stretch}

\centering

\begin{tabular}{ccc}
	\hline\\[-.75em]
	Set Pair  & Vertex Pair & Edge Condition \\[.25em]
	\hline\\[-.5em]
    $A\times A$ & & no edges  \\[.5em]
    $B\times B$ & $b[ \nwspace i,j,k],\ b[ \nwspace x,y,z]$  
    	& $(i=x) \oplus (j=y)$ \\[.5em]
    $C\times C$ & $c[ \nwspace i,j,k],\ c[ \nwspace x,y,z]$ 
    	& $(i=x) \oplus (j=y)$ \\[.5em]
    $D\times D$ & & no edges \\[.5em]
    $A\times B$ & $a[ \nwspace i,j],\ b[ \nwspace x,y,z]$ 
    	& $(i=x) \wedge (z = 0)$\\[.5em]
    $B\times C$ & $b[ \nwspace i,j,k],\ c[ \nwspace x,y,z]$ 
    	& $(i=x) \wedge (j=y)$ \\[.5em]
    $C\times D$ & $c[ \nwspace i,j,k],\ d[ \nwspace x,y]$ 
    	& $(j=y) \wedge (k = 0)$
\end{tabular}
\end{table}

\begin{figure}[b]
\centering
\includegraphics[scale=0.85, page=5]{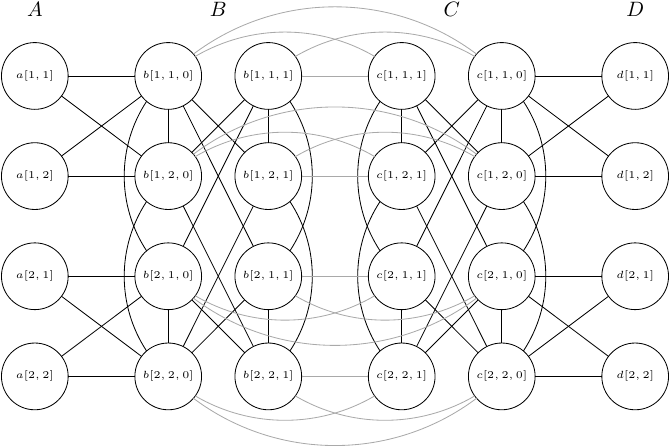}
\caption{%
	The graph $H$ for parameter $N = 4$. The edges between the vertex sets $B$ and $C$
	are drawn in gray to avoid clutter.
}
\label{fig:lower_bound}
\end{figure}

\begin{lemma}
\label{lemma:number_of_edges}
	The graph $H$ has $\Theta(n^{3/2})$ edges.\footnote{
		If one is interested in the leading constants, the proof of \Cref{lemma:number_of_edges} shows 
		that there are $2N$ vertices of degree $5\sqrt{N}{-}2$, $2N$ more of degree $4\sqrt{N}{-}2$,
		and the remaining $2N$ vertices have degree $\sqrt{N}$.
		The handshaking lemma implies that there are 
		$10 N^{3/2}-4N = \frac{5}{3 \sqrt{6}} n^{3/2}- \frac{2}{3}n$ edges.
	} 
\end{lemma}

\begin{proof}[Proof.]
	We claim that the maximum degree in $H$ is $5\sqrt{N}-2$, 
	assumed by the vertices $b[\nwspace i,j,0]$ and $c[\nwspace i,j,0]$ 
	for any $i,j \in [\sqrt{N}\nwspace]$.
	These $2N$ vertices are a constant fraction of $V(H)$.
	The claim thus implies that there are $\Theta(N^{3/2}) = \Theta(n^{3/2})$ edges.
	
	To prove the claim, 
	we fix indices $i^*, j^*$ and calculate the degree of the vertex $v = b[  \nwspace i^*,j^*, 0]$.
	It has no neighbors in $D$ and $\sqrt{N}$ neighbors in $A$, namely, 
	$a[ \nwspace i^*, 1], \dots, a[ \nwspace i^*, \sqrt{N} \nwspace]$.
	The neighbors in $C$ are $c[\nwspace i^*, j^*, 0]$ and $c[\nwspace i^*, j^*, 1]$.
	To count $v$'s neighbors in $B$, first consider the subset that shares its first component,
	$B[\nwspace i^*] = \{b[\nwspace i^*, j,k] \mid j \in [\sqrt{N} \nwspace], k \in \{0,1\}\}$.
	The vertex $v$ is connected to all of those except $b[ \nwspace i^*,j^*,1]$ and $v$ itself,
	which makes for $2 \sqrt{N} - 2$ neighbors.
	Finally, among the vertices in $B{\setminus}B[\nwspace i^*]$ with a different first component,
	$v$ is connected to the set 
	$\{b[\nwspace 1, j^*, k], \dots, b[\nwspace i^*{-}1, j^*, k], b[\nwspace i^*{+}1, j^*, k], \dots,
		b[\sqrt{N}, j^*, k]\}_{k \in \{0,1\}}$,
	contributing the remaining $2(\sqrt{N}{-}1)$ neighbors.
	
	We briefly argue that this is indeed the maximum degree.
	The edges incident to $c[\nwspace i^*, j^*, 0]$ 
	are symmetrical to those of $b[\nwspace i^*, j^*, 0]$.
	The vertex $b[i^*, j^*, 1]$ (resp., $c[\nwspace i^*, j^*, 1]$)
	is missing the $\sqrt{N}$ edges into $A$ (into $D$).
	Also, the vertices in $A$ and $D$ have only degree $\sqrt{N}$ to begin with. 
\end{proof}

\begin{figure}[t]
\centering
\includegraphics{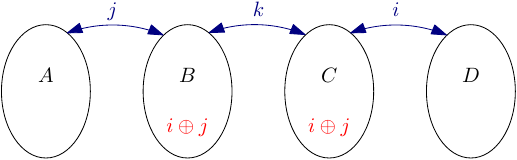}
\caption{%
	Schematic representation of the transitions in graph $H$.
	The label $i$ refers to the first index, $j$ to the second one, and $k$ 
	to the third index (used only in $B$ and $C$).
	The transition from $A$ to $B$ is labeled with a blue $j$ to illustrate that
	an edge $\{a[ \nwspace i,j], b[ \nwspace x,y,0]\}$
	is allowed to jump from any $j$ to any $y$,
	while $i = x$ must remain the same;
	likewise, for the transition from $C$ to $D$ which may vary index $i$.
	Going from $B$ to $C$ must leave both $i$ and $j$ in place, but may change $k$.
	The red labels mark that
	when moving within the sets $B$ or $C$
	either the first or the second index change.
}
\label{fig:schematic_transitions}
\end{figure}

\begin{lemma}
\label{lemma:diam-3}
	The diameter of $H$ is $3$. 
\end{lemma}

\begin{proof}[Proof.]
	We give explicit shortest paths of length at most $3$ between all possible pairs of vertices.
	The edges in $H$ impose rules how the indices of a vertex change when transitioning between the sets
	$A$, $B$, $C$, and $D$.
	\Cref{fig:schematic_transitions} serves as a reference.
	Recall that the edges are undirected, thus all paths below work both ways.
	For some index $i \in [\sqrt{N} \nwspace ]$, we use the symbol $\overline{i}$ 
	for an arbitrary \emph{other} index from $[\sqrt{N} \nwspace ] {\setminus} \{i\}$.
	Since $N \ge 4$, we have $\sqrt{N} \ge 2$ and such an alternative index exists.

	\begin{itemize}
		\item Consider (distinct) vertices $a[ \nwspace i,j],a[ \nwspace x,y]\in A$.
			They are joined by any path of the form
	 		$(a[ \nwspace i,j], \, b[ \nwspace i,y,0], \, b[ \nwspace x,y,0], \, a[ \nwspace x,y])$.
	 		If the respective first indices $i \neq x$ are different, the path has length $3$;
	 		otherwise, the inner vertices $b[ \nwspace i,y,0]$ and $b[ \nwspace x,y,0]$ are the same
	 		and the path shortens to length $2$.\vspace*{.5em}
	 	\item The case $d[ \nwspace i,j],d[ \nwspace x,y]\in D$,
	 		is symmetric to that of two vertices from $A$.\vspace*{.5em}
	 	\item  Two vertices from $B$, or from $C$ respectively, have distance at most $2$.
	 		Let $b[ \nwspace i,j,k],b[ \nwspace x,y,z] \,{\in}\, B$ be two different vertices.
	 		If exactly one out of the equalities $i = x$ and $j = y$ hold,
	 		they have a direct edge between them.
	 		If we have both $i = x$ and $j = y$, 
	 		it must be that the vertices differ exactly the third component $k \neq z$.
	 		Let $\overline{i}$ be any index different from $i$.
	 		The two vertices are joined by the path
	 		$(b[ \nwspace i,j,k], \, b[ \nwspace \overline{i},j,k], \, b[ \nwspace i,j,z])$.	 		
	 		Finally, if $i \neq x$ and $j \neq y$.
			the path is $(b[ \nwspace i,j,k], \, b[ \nwspace x,j,k], \, b[ \nwspace x,y,z])$.
			The argument for vertices  $c[ \nwspace i,j,k],c[ \nwspace x,y,z] \in C$
			is the same.\vspace*{.5em}
		\item For vertex pairs $(a[ \nwspace i,j],b[ \nwspace x,y,z])\in A \times B$,
			the key observation is that any edge inside of $B$ changes exactly 
			one of the first two indices.
			If $i \neq x$, the path is 
			$(a[ \nwspace i,j], \, b[ \nwspace i,y,0], \, b[ \nwspace x,y,z])$.
			Otherwise, we have $i = x$.
			Hence, the target vertex is $b[ \nwspace x,y,z] = b[ \nwspace i,y,z]$
			and the respective path is
			$(a[ \nwspace i,j], \, b[ \nwspace i,\overline{y},0], \, b[ \nwspace i,y,z])$.
			In summary, the distance from $A$ to $B$ is at most $2$.
			The pairs $(d[ \nwspace i,j],c[ \nwspace x,y,z]) \in D \times C$
			are handled symmetrically.\vspace*{.5em}
		\item The vertex pairs $(a[ \nwspace i,j],c[ \nwspace x,y,z]) \,{\in}\, A \,{\times}\, C$
			have shortest paths of the form
			$(a[ \nwspace i,j], \, b[ \nwspace i,y,0], \, c[ \nwspace i,y,z], \, c[ \nwspace x,y,z])$.
			Note that if $i \,{=}\, x$ the last two vertices are the same.
			Same holds for paths from $D$ to $B$.\vspace*{.5em}
		\item  Pairs $(a[ \nwspace i,j],d[ \nwspace x,y]) \,{\in}\, A \,{\times}\, D$: 
			paths $(a[ \nwspace i,j], \, b[ \nwspace i,y,0], \, c[ \nwspace i,y,0], \, d[ \nwspace x,y])$.\vspace*{.5em}
		\item Pairs $(b[ \nwspace i,j,k],c[ \nwspace x,y,z]) \,{\in}\, B \,{\times}\, C$:
			$(b[ \nwspace i,j,k], b[ \nwspace x,j,k], c[ \nwspace x,j,k], c[ \nwspace x,y,z])$,
			possibly shortened if consecutive vertices are the same. \qedhere
	\end{itemize}
\end{proof}

\paragraph*{The construction of graph $G$}

In the next step of the proof of \Cref{lem:lower-bound},
we show how to encode the entries of a large class\footnote{
	The relevant subclass of matrices will turn out to be all
	$M \in \{0,1\}^{\sqrt N \times \sqrt N \times \sqrt N}$
	that are constraint to have entry $M[\nwspace i, j, \ell \nwspace ] = 0$
	whenever $i = j$ or $j = \ell$.
}
of $\sqrt N \,{\times}\, \sqrt N \,{\times}\, \sqrt N$
binary matrix (tensors) in the fault-tolerant diameter of some supergraph $G \supseteq H$.
In other words, let $M \in \{0,1\}^{\sqrt N \times \sqrt N \times \sqrt N}$ be a matrix
and $\ell$ be another index ranging over $[\sqrt{N} \nwspace]$.
We construct a graph $G$ that, for each triple $(i,j,\ell)$, has a set $F$ 
of two edges such that one can infer the entry $M[\nwspace i,j,\ell \nwspace]$
from $\diam(G{-}F)$.
 
The graph $G$ has the same vertex set $V = V(H)$ as $H$.
It also inherits all edges of $H$,
but additionally contains the following edges that depend on $M$.

\begin{itemize}
	\item For all $i,j,\ell \in [\sqrt N \nwspace]$ with $M[ \nwspace i,j,\ell \nwspace]=1$,
		add the edges $\{ a[ \nwspace i,j],b[ \nwspace i,\ell,1]\}$ 
		and $\{c[ \nwspace i,\ell,1],d[ \nwspace j,\ell \nwspace]\}$ to $G$.
\end{itemize}

\noindent
Note that $O(N^{3/2})$ edges are added 
and the diameter of $G$ is also at most $3$.

\subsection{Fault-Tolerant Diameter Revealing Matrix Entries}
\label{subsec:lower_bound_matrix}

Fix four indices $i,j,x,y \in [\sqrt N \nwspace]$ 
such that $i \neq x$ and  $j \neq y$.
We define two sets $F,E^*$, each one containing exactly two pairs of vertices from $V$.
The first one $F$ in fact contains two edges that are present in $H$. 
Let $e_1 = \{ a[ \nwspace i,j], b[ \nwspace i,y,0] \}$ 
and $e_2 = \{ c[ \nwspace i,y,0], d[ \nwspace x,y] \}$
and $F = \{e_1, e_2\}$.
The elements of the second set $E^*$ are only edges of $G$ if
the entries $M[ \nwspace i,j,y]$ and $M[ \nwspace i,x,y]$ are both $1$. 
In more detail, let $E^*$ be the set comprising the two pairs
$e^*_1 = \{ a[ \nwspace i,j],b[ \nwspace i,y,1]\}$ 
and $e^*_2 = \{c[ \nwspace i,y,1],d[ \nwspace x,y]\}$.

The next two lemmas are crucial in showing how the fault-tolerant diameter
of $G$ depends on $M[ \nwspace i,j,y]$ and $M[ \nwspace i,x,y]$.
We would like to highlight the fact that both of them
hold regardless of the other entries in $M$.

\begin{lemma}
\label{lemma:diam-bound-1}
	For any four indices $i,j,x,y \in [\sqrt N \nwspace ]$ such that $i \neq x$ and  $j \neq y$,
	if the edges in $E^*$ are absent from $G$,
	then the diameter of $G - F$ is at least $5$.
\end{lemma}

\begin{proof}[Proof.]
	Suppose none of the edges in
	$E^* = \{\{ a[ \nwspace i,j],b[ \nwspace i,y,1]\}, \{c[ \nwspace i,y,1],d[ \nwspace x,y]\}\}$ 
	are present in $G$.
	We show that the distance between the vertices $a[ \nwspace i,j]$ and $d[ \nwspace x,y]$ in $G - F$
	is at least $5$.
	Observe that, when moving from one vertex to the other,
	both the first and the second index must change.
	The high-level idea of this proof is that we need sufficiently many edges on any path
	in order to accompany those index changes.	
	
	To reach a contradiction, assume that there exists three (not necessarily distinct)
	vertices $u,v,w \in V$ such that
	$P=(a[ \nwspace i,j], \, u, \, v, \, w, \, d[ \nwspace x,y])$
	is a path of length at most $4$. 
	$P$ must pass across sets $A\rightarrow B$, $B\rightarrow C$, and $C\rightarrow D$.
	Recall that the edge $\{ a[ \nwspace i,j], b[ \nwspace i,y,0] \} \in F$ failed in $G-F$.
	Since $E^*$ is absent, the entries $M[\nwspace i, j, y \nwspace]$ and $M[\nwspace i, x, y \nwspace]$
	are both $0$.
	However, other 1-entries of $M$ may have resulted in additional neighbors of $a[ \nwspace i,j]$.
	In summary, the neighborhood $N_{G-F}( a[ \nwspace i,j])$ is
    \begin{equation*}
    	\Big\lbrace b[ \nwspace i,\overline{y},0] 
			\mid \overline{y} \in [\sqrt{N} \nwspace]{\setminus}\{y\} \Big\rbrace         
            \cup \Big\lbrace b[ \nwspace i,\overline{y},1] 
            	\mid \overline{y} \in [\sqrt{N} \nwspace] {\setminus} \{y\}
            	\wedge M[ \nwspace i,j,\overline{y} \nwspace] = 1 \Big\rbrace.
    \end{equation*}
    Therefore, the first edge $\{a[ \nwspace i,j], u\}$ of $P$
	cannot change the index $i$
	and we have $u = b[\nwspace i, \overline{y}, k]$ 
	where $\overline{y}$ is \emph{different} from $y$, and $k$ may be $0$ or $1$ 
	depending on $M$.
	
	Now note that also the edge $\{ c[ \nwspace i,y,0], d[ \nwspace x,y] \} \in F$ failed.
	Applying the same reasoning to the last edge $\{w,d[ \nwspace x,y]\}$ of the path $P$,
	shows that $w$ must lie in $C$ and have the form $w = c[\, \overline{i} , y, z]$,
	where again $z \in \{0,1\}$ depending on the entries of $M$.
	
	At least one of the edges $\{u,v\}$ or $\{v,w\}$ passes from $B$ to $C$.
	Let us first assume that it is $\{u,v\}$.
	This edge must have already been present in $H$
	and cannot change \emph{any} of the first two indices,
	that is, $v = c[\nwspace i, \overline{y}, \xi]$ for some $\xi \in \{0,1\}$.
	This implies that the edge
	$\{v,w\} = \{ c[\nwspace i, \overline{y}, \xi], c[\, \overline{i} , y, z]\}$
	connects two vertices in $C$ and must change the first two indices simultaneously.
	However, by definition all edges within $C$ change exactly one out of the first two.
	
	If instead the edge $\{v,w\}$ passes from $B$ to $C$,
	then a symmetric argument shows that $\{u,v\}$ must change the first two indices
	while staying in $B$, which again is impossible.
	This also covers the case where $P$ has even fewer than $4$ edges
    as still one of those must eventually pass from $B$ to $C$
    and the other ``edge'' degenerates to a single vertex.
\end{proof}

\begin{lemma}
	\label{lemma:diam-bound-2}
	If the edges in $E^*$ are present,
	then the diameter of $G - F$ is $3$.
\end{lemma}

\begin{proof}[Proof.]
	The proof is essentially the same as that of \Cref{lemma:diam-3}.
	The only difference is that whenever the edge 
	$e_1 = \{a[ \nwspace i,j], b[ \nwspace i,y,0]\} \,{\in}\, F$ would be used,
	the respective path instead traverses $e^*_1 = \{a[ \nwspace i,j], b[ \nwspace i,y,1]\} \in E^*$.
	Likewise, every occurrence of 
	$e_2 = \{ c[ \nwspace i,y,0], d[ \nwspace x,y] \}$
	is replaced by $e^*_2 = \{ c[ \nwspace i,y,1], d[ \nwspace x,y] \}$.
\end{proof}

We now finish the proofs\footnote{
	Part of the proof of \Cref{lem:lower-bound} is to reconcile the $i,j,x,y$ index naming scheme 
	from the start of \Cref{subsec:lower_bound_matrix} with the $i,j,\ell$ scheme used for $M$.
	The former is chosen to align the proofs of \Cref{lemma:diam-bound-1,lemma:diam-bound-2}
	with the definition of graph $H$.
	The latter emphasizes the three degrees of freedom when addressing the entries of $M$.
} of \Cref{lem:lower-bound} and, in turn, \Cref{thm:lower-bound}.
Suppose there exists a data structure for the graph $G$ with sensitivity $f \ge 2$ 
that distinguishes whether the fault-tolerant diameter is bounded by $3$ or at least $5$.
We can use it to infer the entry $M[ \nwspace i,j,\ell\nwspace]$
for any triple $(i,j,\ell) \in [\sqrt{N} \nwspace]^3$ of indices 
such that $i$ and $j$ differ from each other, and $j$ and $\ell$ differ.
In the notation used at the start of \Cref{subsec:lower_bound_matrix},
we set $x = j$ and $y = \ell$.
The resulting failure set is therefore
\begin{multline*}
	F = \Big\lbrace \{ a[\nwspace i,j], b[\nwspace i, y, 0]\}, \{c[\nwspace i,y,0], d[ \nwspace x,y]\}
		\Big\rbrace\\
		= \Big\lbrace \{ a[\nwspace i,j], b[\nwspace i, \ell, 0]\}, 
			\{c[\nwspace i, \ell,0], d[ \nwspace j,\ell \nwspace ]\} \Big\rbrace
\end{multline*}

\noindent
Our assumptions $i \,{\neq}\, j$ and $j \,{\neq}\, \ell$
ensure that $i \,{\neq}\, x$ and $j \, {\neq}\, y$,
the \Cref{lemma:diam-bound-1,lemma:diam-bound-2} thus apply.
If $M[\nwspace i,j,\ell \nwspace] = 1$,
then $M[i,j,y] = M[i,x,y] = 1$ since they are all the same entry,
and the respective edge set $E^*$ is present in $G$.
\Cref{lemma:diam-bound-2} implies that $\diam(G{-}F) \ {=}\ 3$.
If conversely $M[\nwspace i,j,\ell \nwspace] = 0$,
then $M[i,j,y] = M[i,x,y] = 0$, $E^*$ is absent, and  $\diam(G{-}F) \ge 5$.
For the assertion of \Cref{lem:lower-bound} about the $ST$-diameter (instead of the unrestricted diameter),
we just choose $S = A$ and $T=D$.

There are $2^{\sqrt{N} (\sqrt{N}-1)^2} =  2^{\Omega(N^{3/2})} = 2^{\Omega(n^{3/2})}$ 
possibilities of fixing the entries in $M$ whose indices $(i,j,\ell) \in [\sqrt{N} \nwspace]^3$ 
satisfy $i \neq j$ and $j \neq \ell$.
The data structure in question must therefore take $\Omega(n^{3/2})$ bits of space.
The proof of \Cref{thm:lower-bound} is completed by observing that any $f$-FDO-$ST$
with stretch better than $5/3$ is indeed capable of distinguishing
those two cases.

\subsection*{Acknowledgments}

The authors thank the anonymous Algorithmica reviewers 
for their many helpful comments, and for sharing constructions
that improved the space overhead, query time, and preprocessing time 
of \Cref{thm:st-diameter-where-st-are-sets}.

\subsection*{Statement of Competing Interests}

The authors have no relevant financial or non-financial interests to disclose.

\bibliography{algorithmica_bib}

\end{document}